%% file: main_cikm_arxiv.tex
\documentclass[sigconf]{acmart}

\usepackage{pgfplots}
\pgfplotsset{compat=1.9}

\usepackage{hyperref}
\usepackage{subcaption}
\usepackage{amsthm}
\usepackage{amsmath}
\usepackage{siunitx}
\usepackage{makecell}
\usepackage{tablefootnote}
\usepackage{tikz} 

\usepackage{algorithm}
\usepackage{algpseudocode}

\newtheorem{proposition}{Proposition}[section]
\newtheorem{lemma}{Lemma}[section]

\usepackage{tabstackengine}
\setstackEOL{\cr}

\AtBeginDocument{%
  \providecommand\BibTeX{{%
    \normalfont B\kern-0.5em{\scshape i\kern-0.25em b}\kern-0.8em\TeX}}}

\copyrightyear{2021}
\acmYear{2021}
\setcopyright{acmcopyright}
\acmConference[CIKM '21] {Proceedings of the 30th ACM International Conference on Information and Knowledge Management}{November 1--5, 2021}{Virtual Event, Australia.}
\acmBooktitle{Proceedings of the 30th ACM Int'l Conf. on Information and Knowledge Management (CIKM '21), November 1--5, 2021, Virtual Event, Australia}
\acmPrice{15.00}
\acmISBN{978-1-4503-8446-9/21/11}
\acmDOI{10.1145/XXXXXX.XXXXXX}

\begin{document}

\fancyhead{}

\title{Grassland: A Rapid Algebraic Modeling System for Million-variable Optimization}

\author{Xihan Li}
\affiliation{%
  \institution{University College London}
  \city{London}
  \country{The United Kingdom}}
\email{xihan.li@cs.ucl.ac.uk}
\author{Xiongwei Han}
\affiliation{%
  \institution{Huawei Noah's Ark Lab}
  \city{Shenzhen}
  \country{China}}
\email{hanxiongwei@huawei.com}
\author{Zhishuo Zhou}
\affiliation{%
  \institution{Fudan University}
  \city{Shanghai}
  \country{China}}
\email{zhouzs18@fudan.edu.cn}
\author{Mingxuan Yuan}
\affiliation{%
  \institution{Huawei Noah's Ark Lab}
  \city{Shenzhen}
  \country{China}}
\email{yuan.mingxuan@huawei.com}
\author{Jia Zeng}
\affiliation{%
  \institution{Huawei Noah's Ark Lab}
  \city{Shenzhen}
  \country{China}}
\email{zeng.jia@huawei.com}
\author{Jun Wang}
\affiliation{%
  \institution{University College London}
  \city{London}
  \country{The United Kingdom}}
\email{jun.wang@cs.ucl.ac.uk}


\begin{abstract}
  An algebraic modeling system (AMS) is a type of mathematical software for optimization problems, which allows users to define symbolic mathematical models in a specific language, instantiate them with given source of data, and solve them with the aid of external solver engines. With the bursting scale of business models and increasing need for timeliness, traditional AMSs are not sufficient to meet the following industry needs: 1) million-variable models need to be instantiated from raw data very efficiently$;$ 2) Strictly feasible solution of million-variable models need to be delivered in a rapid manner to make up-to-date decisions against highly dynamic environments. \textit{Grassland} is a rapid AMS that provides an end-to-end solution to tackle these emerged new challenges. It integrates a parallelized instantiation scheme for large-scale linear constraints, and a sequential decomposition method that accelerates model solving exponentially with an acceptable loss of optimality. Extensive benchmarks on both classical models and real enterprise scenario demonstrate $6\sim10$x speedup of Grassland over state-of-the-art solutions on model instantiation. Our proposed system has been deployed in the large-scale real production planning scenario of Huawei. With the aid of our decomposition method, Grassland successfully accelerated Huawei's million-variable production planning simulation pipeline from hours to $3 \sim 5$ minutes, supporting near-real-time production plan decision making against highly dynamic supply-demand environment.
\end{abstract}




\maketitle

\section{Introduction}

Mathematical optimization is a powerful analytics technology that allows companies to make optimal decisions based on available business data, widely applied in industry scenarios including logistics\cite{epstein2012strategic}, manufacturing\cite{chen1997linear}, finance\cite{cornuejols2006optimization} and energy. It abstracts key features of a complex business problem as an optimization model, which consists of objectives (business goal), variables (decisions to be made) and constraints (business rules). In such a way the business problem is decomposed into two stages: \textit{conversion} --- converting the problem to a canonical optimization model, and \textit{solving} --- finding the optimal or approximated solution of the model. 

In practice, the conversion stage consists of two steps. The first step is \textit{modeling}, which means writing down the expression of objective and constraints in a formulated way. E.g., using mathematical expression $\sum_{i \in S} x_i \leq cap$ to formulate a capacity constraint that the total production amount of specific products should not exceed the plant capacity. The second step is \textit{instantiation}, which means to generate a particular instance of the model when real data are available. E.g., when we know today's plant capacity is $cap = 100$ and available set of products is $S = \{1, 3, 5\}$, we get a particular instance of expression $x_1 + x_3 + x_5 \leq 100$ for today (and tomorrow's instance might be very different). For the solving stage, we usually use solver engines that are sophisticatedly developed to solve specific kinds of models, such as Gurobi, CPLEX, Mosek and CLP. Decomposition methods may also apply when the model is large. A toy example of mathematical optimization pipelines for practical online business scenarios is shown in \autoref{fig:MP_pipeline}, in which the high-level system to handle both the conversion and solving stages is usually called \textit{algebraic modeling system} (AMS).

In recent years, the scale and complexity of business problems are dramatically increased, e.g., a large electronics company can involve over $ 10^5 $ types of products and $ 10^2 $ plants worldwide with different standard. As a result, their corresponding optimization models become extremely massive and cumbersome. They not only involve millions of decision variables, but also contain extremely lengthy real-world business constraints, which cover numerous case of real business logic (e.g., production hierarchy, inventory control, supply-demand modeling, delay minimization, intra/inter-factory transshipment and replacement) thus can take hundreds of pages to document. While million-variable models can be a burden for solving stage, instantiating hundred-page business constraints efficiently as standard models is also highly nontrivial for conversion stage. For end-to-end optimization, both the two stages can be extremely time-consuming with traditional toolchain.

However, the information age calls for \textit{rapid optimization}. 
Only in this way can business decisions be frequently adjusted and updated to reflect the latest state of fast-changing market, and fulfill customers' increasing need for timeliness. The more rapidly getting optimized decisions from latest available data, the more timely the company can respond to market change and other uncertainty factors. 
This especially applies to rolling horizon that a time-dependent model is solved repeatedly. E.g., a manufacturing company that can update its production plan in minutes (high re-planning periodicity) to handle unexpected urgent orders or factory stoppage is more competitive than its counterparts who have to wait for hours (low re-planning periodicity) for a new optimized plan. The same goes for other scenarios like logistics and finance, in which the timeliness of business decisions is directly related to user experience and profits, acting as a core competitiveness of modern enterprises. 

Moreover, rapid optimization creates remarkable possibilities for business intelligence. First, it can act as an analytical tool that provides metrics for high-level business decisions. E.g., buyers can evaluate different raw material purchase plans via running a ``simulated'' production planning model for each plan, and check their corresponding order fulfillment rates. 
Second, it can serve as a basis for larger or more complicated optimization. E.g., to strictly conserve order priority constraint (which is nonlinear), we can run production planning multiple times, where high-priority orders are planned before low-priority ones. All these build on the cornerstone that a single shot of end-to-end optimization can be very rapid.



\begin{figure}
	\centering
	\includegraphics[width=\linewidth]{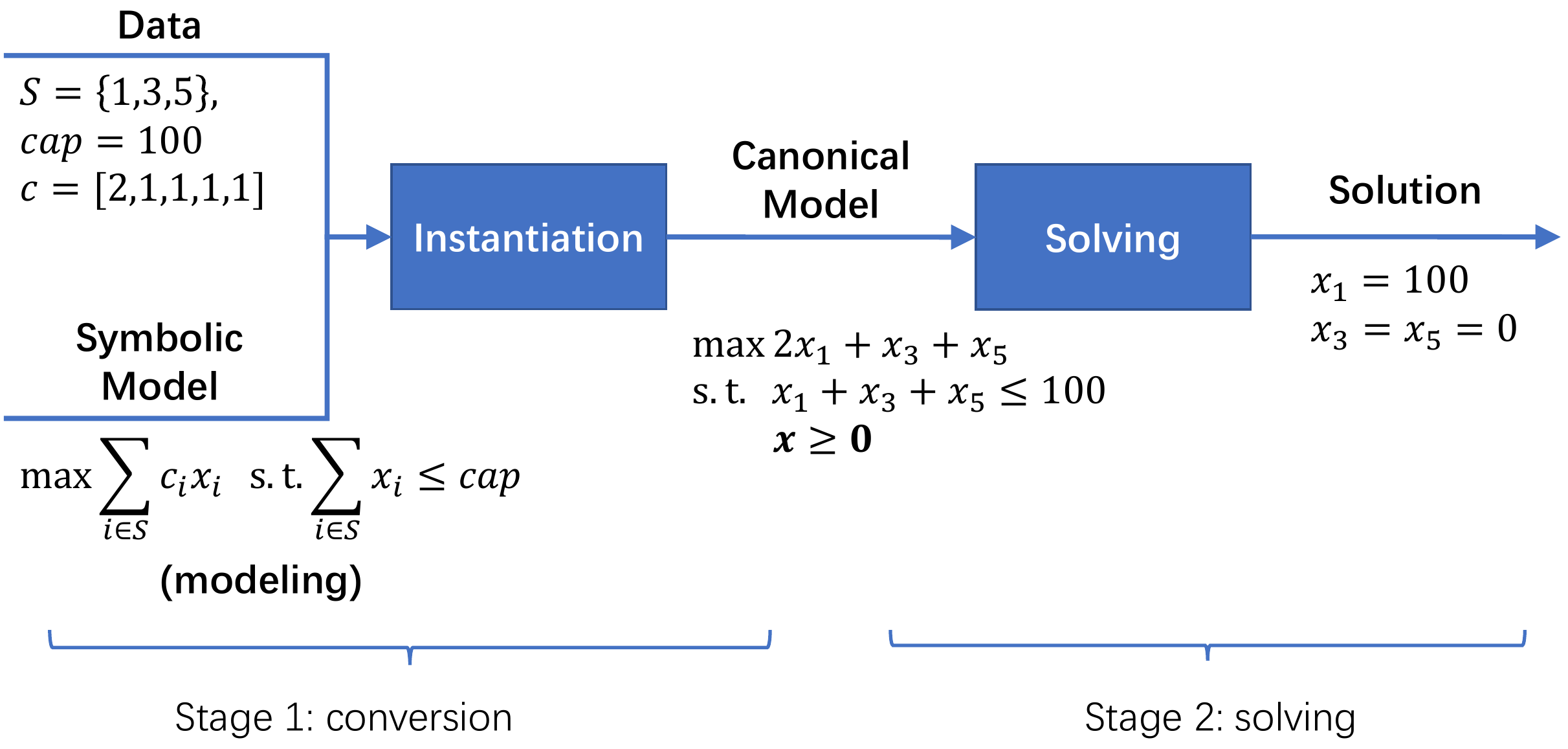}
	\vspace{-1em}
	\caption{A toy example of mathematical optimization pipeline for practical business decision making scenarios. 
	}
	\vspace{-1em}
	\label{fig:MP_pipeline}
\end{figure}

Sadly, current state-of-the-art AMSs are far from supporting the aforementioned ambition of rapid optimization. In \autoref{fig:MP_pipeline}, for the conversion stage, they lack a principled design to stress the efficiency of model instantiation, especially ignoring the parallelization and vectorization of operations
. This issue is minor for normal-sized, simple models, but significantly raised as a major bottleneck for million-variable, hundred-page documented business scenarios. For the solving stage, they focus on lossless decomposition methods such as Benders decomposition, which are not widely applicable since the special block structures they required are easily broken in complex real scenarios. 
As a result, both two stages in \autoref{fig:MP_pipeline} are desperately time-consuming (typically several hours) in large-scale real scenarios, hindering companies from building highly-responsive decision systems against fast-changing markets.

To achieve the ambition of rapid end-to-end optimization, both the performance bottleneck of instantiation and solving must be removed. In this paper, we propose two methods that fully address the efficiency of the two stages respectively, and encapsulate them as a new AMS, \textit{Grassland}. For model instantiation, the motivation is to take advantage of both the sparsity of the data and the modern multiprocessor systems. While multiprocessor parallelism usually accompanies with an extra cost of communication and synchronization, we eliminate such a cost by a specially vectorized formulation that fully exploits the parallelism of data. In such a way we developed a model instantiation algorithm that is not only comparable with state-of-art AMSs in a single-threaded setting, but can also be accelerated in direct proportion to the number of processor cores. For model solving, we start from rolling horizon, a common business practice for decision making, to decompose a full model into a sequence of smaller models
. While such a decomposition can lead to a significant loss of global optimality, we propose a heuristic of additional ``aggregated master problem'' to capture global optimality, so as to minimize the loss. We also integrate our proposed heuristics with an existing one to further improve the performance.

Our major contributions can be summarized as follows:

\begin{itemize}
	\item Proposing a principled approach to stress the efficiency of model instantiation for optimization problems with linear constraints, by exploiting both the sparsity and parallelism of data, represented as a new algebraic modeling language with corresponding instantiation algorithms.
	\item Proposing Guided Rolling Horizon, a new decomposition heuristics that accelerates solving exponentially for mathematical optimization models with sequential structure with an acceptable loss of optimality, which can also work with other heuristics (Guided FRH) for better performance.
	\item Encapsulating our approaches as a new AMS, \textit{Grassland}, and conducting extensive experiments on both classical LP/MIP models and real enterprise production planning scenarios.
\end{itemize}

This system has been deployed in Huawei's supply chain management scenario, especially for production planning simulation of our planning department. By reducing the end-to-end optimization pipeline from hours to 3-5 minutes, Grassland achieves our ambition of rapid optimization in real industrial scenarios, playing an essential role in the production plan decision making of Huawei against highly dynamic supply-demand environment. 

\section{Related Work}

Algebraic modeling systems/languages is a relatively mature field starting from the late 1970s, with many sophisticatedly designed open-source and commercial software available. They can be divided into three categories. 1) Classical standalone systems with particular modeling language syntax, such as AMPL\cite{fourer1990modeling}, GAMS\cite{brook1988gams} and ZIMPL\cite{Koch2004}. They are usually more efficiently developed and easier to use for non-programmers; 2) Standalone modeling packages based on specific programming languages, such as YALMIP\cite{lofberg2004yalmip} for MATLAB, Pyomo\cite{hart2011pyomo} for Python, and JuMP\cite{LubinDunningIJOC} for Julia. 
Their modeling efficiency usually lies on the host language. 3) The modeling API as an integrated part of several solver engines, such as Gurobi\cite{gurobi2018gurobi} and Mosek\cite{mosek}, which usually provides solver-specific features. 

For the model instantiation process, while most of the aforementioned AMSs stress the universality of modeling, hardly any of them take the efficiency issue (e.g., parallelization and vectorization) into full consideration, especially for large-scale scenarios with millions of variables. One may argue that such efficiency issue could be leaved to practitioners. However, it is not practical for hundred-page documented complex optimization models to be manually analyzed line-by-line for instantiation efficiency, as if manual calculation of gradients is not practical for complex deep learning models (although neither of them contains theoretical difficulties!). For complex models, principled design must be given to achieve practical implementation efficiency. While machine learning communities benefit enormously from such principledly designed frameworks such as TensorFlow\cite{tensorflow} and PyTorch\cite{PyTorch} , industrial optimization practitioners call for an analogous framework that boost the instantiation of highly complex business optimization models. 

For decomposition methods of large-scale model solving, most of the current literature focuses on lossless decomposition such as Benders decomposition or column generation\cite{BradleyStephenP.1977Amp}, in which the global optimality is guaranteed. However, such methods require a special block structure that is not easily satisfied in realistic complex scenarios. It is common that one or more types of constraints break the block structure, and blind use of decomposition methods will usually lead to performance degradation or even infeasibility. However, compared with feasibility, strict global optimality is not so crucial in most of the application scenarios. A fast-generated good solution is usually more appealing than a time-exhausting perfect one, which is the application foundation of heuristics methods. Existing lossy decomposition methods like forward rolling horizon \cite{DIMITRIADIS1997S1061} include a simple heuristics that aggregate future information to speedup the solving. In this way strict global optimality is slightly satisfied to exchange for more potential of acceleration. A detailed introduction is provided in \autoref{sec:FRH}.

\section{Efficient Instantiation of Linear Constraints in Optimization}

To design a rapid AMS, the first challenge is the efficiency of model instantiation for the emerging large-scale applied scenarios. In this section, we describe the model instantiation problem of mathematical optimization, and propose a general scheme to instantiate large-scale linear constraints efficiently, with both sparsity and parallelism taken into account.

\subsection{Preliminaries}\label{sec:preliminaries}

In this section, we give a brief introduction to mathematical optimization. A mathematical optimization (programming) problem can be represented in the following way: given a function $ f: S \rightarrow \mathbb{R} $ from a set $ S $ to real numbers, find an element $ x_0 \in S $ so that $ \forall x \in S $, $ f(x_0) \leq f(x) $ (minimization) or $ f(x_0) \geq f(x) $ (maximization). Here $ f $, $ S $ and $ x \in S $ are called objective function, feasible region and constraint respectively. 

For the majority of optimization types in practice that can be solved efficiently, $ S $ is a convex polytope. That is, to optimize an objective function subject to linear equality and inequality constraints, which can be expressed in canonical form as
\begin{align}
    \begin{split}\label{eq:lp}
    \min\quad & f(\textbf{x}) \\
    \text{subject to}\quad & \textbf{A}\textbf{x} = \textbf{b}, \textbf{x} \geq \textbf{0}
    \end{split}
\end{align}
in which $ A $ is called constraint matrix. $f(\textbf{x}) = \textbf{c}^T\textbf{x}$ for linear programming (LP) and $f(x) = \frac{1}{2}\textbf{x}^T\textbf{Q}\textbf{x} + \textbf{c}^T\textbf{x}$ for quadratic programming (QP). Mixed integer programming (MIP) adds additional integer constraints $x_i \in \mathbb{Z}, i \in Z$ on LP. Some special types of optimization contains nonlinear constraints such as conic and semidefinite optimization, which are out of this paper's scope.

Practically, while objective function $f(\textbf{x})$ describe the business goal which is only a single expression, constraint matrix $\textbf{A}$ can represent numerous business rules which contain millions of expressions that are much more time-consuming to instantiate. Therefore, we focus more on the efficiency of constraints instantiation in the following text.

\subsection{Problem description and challenges}

In this section, we describe the model instantiation problem of mathematical optimization, as well as its challenges in large-scale scenarios. 

\subsubsection{Problem description}

While solver engines take canonical mathematical optimization problems as input, end users rarely write canonical problems directly. Instead, they develop \textit{symbolic representations of models} in a human-readable language, which is called algebraic modeling language (AML). A symbolic representation of a model is a template of model without any concrete data. The place where concrete data should exist is represented by \textit{placeholders}. When a symbolic model needs to be solved with given data, both the data and AML-based symbolic model are fed into AMS. AMS will compile the symbolic model, and fill the model with concrete data to generate a canonical representation for the solver engine. We name this process as \textit{model instantiation} and conclude the input/output of this process as follows

\begin{description}
	\item[Input] (1) Symbolic representation of the model and (2) data.
	\item[Output] Canonical representation of the model. 
\end{description}

To show the procedure of model instantiation, we show an example of a simplified classical minimum-cost flow model. Consider a direct graph with a set $V$ of nodes and a set $E$ of edges, decision variable $ x_{i, j} $ represents the amount of current flowing from node $ i $ to node $ j $. $s_i$ is the supply/demand at each node $i$. For each node $i$, flow out $\sum_{\{j | (i, j) \in E\}} x_{i, j}$ minus flow in $\sum_{\{j | (j, i) \in E\}} x_{j, i}$ must equal the supply/demand $s_i$. Every flow corresponds to a cost $c_{i,j}$, and the model finds flows that minimize the total cost. The input/output for this model's instantiation is as follows

\setstacktabbedgap{5pt}

\textbf{Input:}

\begin{itemize}
    \item Symbolic representation of the model: \begin{align}\label{eq:network_flow}
        \min &\sum_{(i,j) \in E} c_{i,j}x_{i,j} \nonumber \\
        \text{subject to}\quad expr_i &= s_i, \quad \forall i \in V , \quad \textbf{x} \geq \textbf{0}\\
        \text{in which}\quad expr_i &= \sum_{\{j | (i, j) \in E\}} x_{i, j} - \sum_{\{j | (j, i) \in E\}} x_{j, i}\label{eq:expr_i}
    \end{align}\\
    in which $i$, $j$ are index placeholders denoting the index of expression $expr$ and variable $x$. $V$, $E$ and $S$ are data placeholders whose value need to be specified in model instantiation process.
    \item Data: $ V^* = \{1, 2, 3, 4\}, E^* = \{(1, 2), (2, 4), (1, 3), (3, 4)\}, s^* = [1, 0, 0, -1]$, $c^* = \textbf{1}$
\end{itemize}

\textbf{Output:}
\begin{itemize}
    \item Canonical representation of the model: \eqref{eq:lp} in which
    \begin{align*}
        \textbf{c} &= [0, 0, 0, 0, 1, 0, 0, 0, 1, 0, 0, 0, 0, 1, 1, 0]^T\\
        \textbf{A} &= \bracketMatrixstack{
    0 & 0 & 0 & 0 & 1 & 0 & 0 & 0 & 1 & 0 & 0 & 0 & 0 & 0 & 0 & 0 \cr
    0 & 0 & 0 & 0 & -1 & 0 & 0 & 0 & 0 & 0 & 0 & 0 & 0 & 1 & 0 & 0 \cr
    0 & 0 & 0 & 0 & 0 & 0 & 0 & 0 & -1 & 0 & 0 & 0 & 0 & 0 & 1 & 0 \cr
    0 & 0 & 0 & 0 & 0 & 0 & 0 & 0 & 0 & 0 & 0 & 0 & 0 & -1 & -1 & 0} \\
        \textbf{b} &= [1, 0, 0, -1]^T \\
        \textbf{x} &= [x_{1, 1}, \cdots, x_{4,1}, \cdots, x_{1, 4}, \cdots, x_{4, 4}]^T
    \end{align*}
\end{itemize}

Feed $\textbf{c}$, $\textbf{A}$ and $\textbf{b}$ into a LP solver engine and we can get the solution of decision variable $\textbf{x}$. We will discuss the example output's generation process later in \autoref{sec:instantiation_alg}.

\subsubsection{Challenges of model instantiation}

While model instantiation seems trivial from a theoretical perspective (it can be easily achieved in polynomial time compared with model solving), the challenge is that, the time complexity can still be extremely high if we directly follow the literal meaning of mathematical expressions to instantiate models, especially for constraints. Consider the following general equality constraint
\begin{equation*}
    expr_{i_1, \cdots, i_N} = s_{i_1, \cdots, i_n}, \quad \forall i_1 \in V_1, \cdots, i_N \in V_N
\end{equation*}
The time complexity of direct instantiation is $O(|V_1 \times \cdots \times V_N||expr|)$ in which $|expr|$ is the computation cost to instantiate a single expression. For example, the time complexity of constraint \eqref{eq:network_flow}'s direct instantiation is $ O(|V||E|) $ since by definition, we need to iterate every node $ i $ ($\forall i \in V$ in \eqref{eq:network_flow}), and search for edges whose heads or tails are equal to $ i $ (sum operation $\sum_{\{j | (i, j) \in E\}}$ and $\sum_{\{j | (j, i) \in E\}}$ in \eqref{eq:expr_i}). This is clearly unbearable if we have a sparse graph with millions of vertices and edges.

In application scenarios, such a lack of principled, efficient model instantiation scheme results in serious performance issues. A surprising fact is that model instantiation costs similar or even more time than model solving in many large-scale, complex enterprise scenarios. While ad-hoc solutions may exist for specific kind of problems\footnote{For example, an ad-hoc solution for instantiate Expression \ref{eq:expr_i}'s instantiation is to pre-index the edge by both head and tail node.}, our aim is to develop an AMS which is efficient for general linear model instantiation, whose model definition can be in arbitrary forms.

\subsection{A principled scheme of efficient linear constraint instantiation} \label{sec:model_instantiation}

In this section, we propose a principled scheme to instantiate linear constraints efficiently. The motivation is to design an AML with a corresponding instantiation algorithm that fully exploit the sparsity and parallelism of data. 

\begin{figure}
	\centering
	\includegraphics[width=\linewidth]{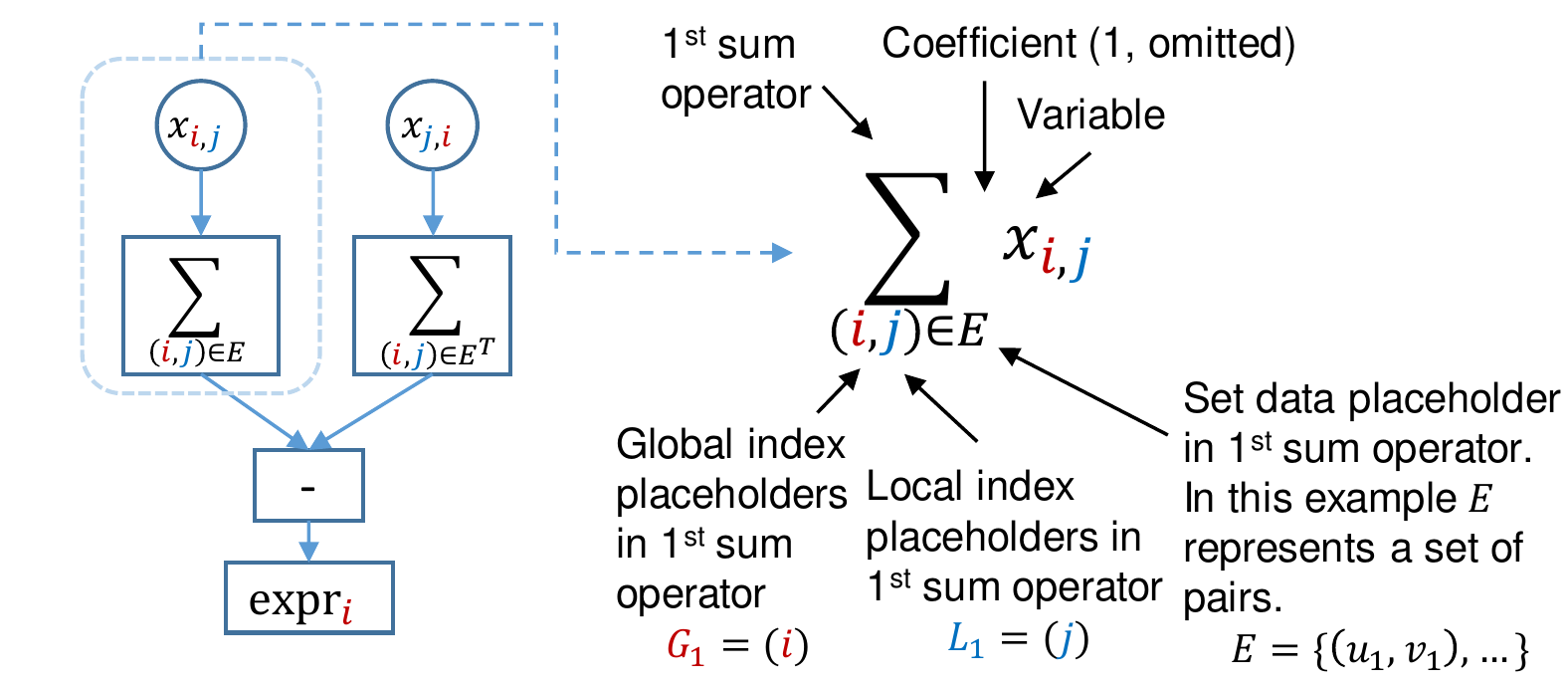}
	\vspace{-1em}
	\caption{The expression tree of expression \eqref{eq:expr_i}, with global index placeholder $ i $ and local index placeholder $ j $.}
	\vspace{-1em}
	\label{fig:multiindex_expression}
\end{figure}

\subsubsection{A new AML with symbolic multidimensional expression}

To represent a model symbolically with AML, the main task is to develop a symbolic representation of mathematical expressions in objective function and constraints. In our system, for a model with $K$ \textit{types} of constraints (e.g., the balance constraint \eqref{eq:expr_i} is one type of constraint. A typical applied model can contain dozens of types of constraints), we formulate the symbolic representation as the following format:
\begin{align}
    \begin{split}\label{eq:symbolic_model}
    \min \quad & expr^o_{G_o} \\
    \text{subject to}\quad & (expr^i_G, sign_i, rhs^i_G), 1 \leq i \leq K
    \end{split}
\end{align}
in which the $i$th type of constraint is formulated as a triple $ (expr_G, \allowbreak\text{sign}, \allowbreak rhs_G) $. That is, a symbolic expression $ expr_G $, a sign ($ = , \geq \text{or} \leq $), and a symbolic array $ rhs_G $. For example, \eqref{eq:network_flow} can be represented as a triple $(expr_G, =, s_G)$ in which $G = (i)$. Here we assume that all constants are moved to the right hand side of the equation, and omit the superscript $i$ of $ expr^i_G $ for simpler notation. 


It is important to notice that the symbolic expression $expr_G$ is \textit{multidimensional}. That is, it represents a list of indexed expressions instead of a single one. For example, $expr_i$ in \eqref{eq:expr_i} actually represents a list of expressions $[expr_1, \cdots, expr_{|V|}]$. Therefore, while a normal symbolic expression $expr$ can be represented as an expression tree whose non-leaf and leaf nodes are operators and terms respectively, A multidimensional expression should be added with additional \textit{index placeholders} $ G $ to denote the index of expressions, represented as $ expr_G $. It consists of the following elements

\begin{itemize}
    \item \textit{(Global) index placeholders} $ G = (g_1, \cdots, g_N) $, a tuple of one or more placeholders representing the index of expressions, defined at the root of the expression tree. Global index placeholders can appear anywhere in the expression tree. e.g., $G = (i)$ in \eqref{eq:expr_i}, and the index placeholder $i$ also appears in the logical condition of two sum operators in \eqref{eq:expr_i}. Here the bracket in $G = (i)$ indicates that $G$ is a tuple even if $N = 1$.
    \item \textit{Operators} (non-leaf nodes), receiving one or more nodes as operands. Some specific operators such as sum ($\sum$) include \textit{local index placeholders} that are only valid in the scope of these operators. E.g. the sum operator $\sum_{\{j | (i, j) \in E\}}$ in \eqref{eq:expr_i} receives $x_{i,j}$ as an operand, and include local index placeholder $j$ that is only valid in the scope of this sum operator. We will discuss the sum operator in detail later. 
    \item \textit{Terms} (leaf nodes), a variable with its corresponding coefficient. Their indices can be denoted by previously defined global and local index placeholders. e.g. $x_{i,j}$ in \eqref{eq:expr_i} is a term with variable $x_{i,j}$ and coefficient 1.
\end{itemize}
and a example is shown in the left part of \autoref{fig:multiindex_expression}.

In this work we mainly focus on the linear constrainted scenario that includes three operators: add ($+$), subtract ($-$) and sum ($\sum$). While the add and subtract operators are trivial, we discuss the sum operator in detail. In our formulation, we number all sum operators in a expression, and sum operators take a fixed format: the $ k $th sum operator is assigned to a symbolic logical condition $ (G_k \| L_k) \in S_k $.\footnote{Here we simplify the denotation $\{L_k | (G_k \| L_k) \in S_k\}$ as $ (G_k \| L_k) \in S_k $} ``$ \| $'' denotes the concatenation of tuples. For example, the 1st ($k = 1$) sum operator $\sum_{(i, j) \in E}$ in \eqref{eq:expr_i} contains a logical condition $(i, j) \in E$. It consists of

\begin{itemize}
    \item \textit{Global index placeholders} $ G_k = (g_{k_1}, g_{k_2}, \cdots) $, a tuple whose elements are the index placeholders from $ G $. e.g., $G_1 = (i)$ in $\sum_{(i, j) \in E}$.
    \item \textit{Local index placeholders} $ L_k = (l_1, l_2, \cdots) $, which is only valid in the scope of this sum operator. e.g., $L_1 = (j)$ in $\sum_{(i, j) \in E}$.
    \item \textit{Data placeholder} $ S_k $ : A symbolic placeholder representing a set of fixed-sized index tuples. The size of each tuple is $ |G_k| + |L_k| $. e.g., edge set placeholder $E$ in $\sum_{(i, j) \in E}$, representing a set of direct edges (i.e., pair of node indices). 
\end{itemize}

An example is shown in the right part of \autoref{fig:multiindex_expression}.


\subsubsection{An efficient model instantiation algorithm}\label{sec:instantiation_alg}

To instantiate a multidimensional expression $ expr_G $ given data $ S^* $, a simple way is to enumerate all possible combinations of global index placeholders $ G $ (denoted as $ \text{space}(G) $), and traverse through the expression tree for each combination to generate every single algebraic expression. This is usually the literal meaning of mathematical expressions. The detailed process is illustrated in \autoref{alg:exhaustive_model_instantiation}. For example, to instantiate \eqref{eq:expr_i} with given data $V^*$ and $E^*$, this expression's index placeholders is $G = (i)$, and all possible values of $i$ are $\text{space}(G) = \{1, 2, 3, 4\}$. Then we iterate all elements of $\text{space}(G)$. e.g., when $i = 1$, we iterate $E$ to generate sub-expression $\sum_{(1, j) \in E} x_{1, j} = x_{1,2} + x_{1,3}$ and $\sum_{(j, 1) \in E} x_{j, 1} = 0$, then the expression will be $expr_1 = x_{1, 2} + x_{1, 3}$. In a similar way we get $expr_2 = x_{2, 4} - x_{1, 2}, expr_3 = x_{3, 4} - x_{1, 3}, expr_4 = -x_{2, 4} - x_{3, 4}$.

However, this approach is extremely exhaustive with time complexity $ O(|\text{space}(G)||S^*|) $. In this section, we propose an efficient model instantiation algorithm based on the AML proposed in previous section, whose result is identical to \autoref{alg:exhaustive_model_instantiation} but with $ |O(S^*)| $ time complexity.

\begin{algorithm}
    \footnotesize
	\caption{Exhaustive model instantiation algorithm}
	\label{alg:exhaustive_model_instantiation}
	\begin{algorithmic}[1]
		\Require Symbolic multidimensional expression $ expr_G $, set data $ S^* $
		\Ensure Constraint matrix $ A $
		\State Initialize $A$ as an empty matrix, with number of columns equal to number of variables.
		\For {$ G^* \in \text{space}(G) $}
			\State Generate $expr^*_G$ by replacing symbolic index placeholders $ G $ and data placeholder $ S $ with concrete value $ G^* $ and $ S^* $ respectively in all nodes of $ expr_G $
			\State Do an in-order traversal to $expr^*_G$ to generate the algebraic expression. When visiting the $i$th sum operator, traverse through corresponding set data $S_i^*$ in its logical condition.
			\State Add one row to $ A $ with the generated expression.
		\EndFor\\
		\Return $ A $
	\end{algorithmic}
\end{algorithm}

\begin{lemma}\label{thm:lemma1}
	For expression tree of $ expr_G $, without loss of generality, we assume that every path from a leaf node to the root will go through at least one sum operator.
\end{lemma}

\begin{proof}
	For the leaf node whose path to the root does not include any sum operator, we can insert a ``dummy'' sum operator $ \sum_{G \in S_i} $ before the node with set data $ S_i^* = \text{space}(G) $. This operator does not contain any local index placeholders so will not change the result of \autoref{alg:exhaustive_model_instantiation}.
\end{proof}

\begin{lemma}\label{thm:lemma2}
	Let $ I(node) = \{i | i\text{th sum operator is on the path} \allowbreak \text{ from} \allowbreak \text{node to root}\} $. Without loss of generality, we assume that for every leaf node of the expression tree, $ \{g | g \in G_j, j \in I(node)\} = G $.
\end{lemma}

\begin{proof}
	From \autoref{thm:lemma1} we know that $I(node) \neq \emptyset$. If $ \exists g' \in G $ so that $ g' \notin \{g | g \in G_j, j \in I(node)\} $, we select one sum operator $ \sum_{(G_i \| L_i) \in S_i} $ on $ I(node) $ and expand $ G_i $ to $ G_i \| g' $. For set data $ S^*_i $, we replace each tuple data $ (g^*_{i_1}, \cdots, g^*_{i_N}, l^*_{i_1}, \cdots) $ to a set of expanded tuple $ \{(g^*_{i_1}, \cdots, g^*_{i_N}, g^*_j, l^*_{i_1}, \cdots) | g^*_j \in \text{space}(g_j)\} $. In this way we enumerate all possible value of $ g_j $ for every tuple data in $ S^*_i $, so will not change the result of \autoref{alg:exhaustive_model_instantiation}.
\end{proof}

With \autoref{thm:lemma1} and \autoref{thm:lemma2}, we propose a model instantiation algorithm. Different from \autoref{alg:exhaustive_model_instantiation} that fixes the value of all global index placeholders and traverses the expression tree for $|\text{space}(G)|$ times, this algorithm traverses the expression tree only once, and records the corresponding value of index placeholders dynamically as ``context information'' when traverse through concrete set data $ S^*_i $ of $i$th sum operator. When the leaf node (term) is reached, all the index placeholders in the term is replaced by the actual value recorded in the context information. Meanwhile, the actual value of all global index placeholders $G$ in the context is snapshotted and attached to the index-replaced term. When the traverse process is finished, we aggregate terms with the same global index. The detailed algorithm is shown in \autoref{alg:efficient_model_instantiation}.

For example, to instantiate \eqref{eq:expr_i}, we traverse the expression tree of \eqref{eq:expr_i} in \autoref{fig:multiindex_expression}. When we arrive at the first sum operation $\sum_{(i, j) \in E}$, we iterate the value $(i, j) \in E$ and record the context information (e.g., record $i = 1, j = 2$ for the first edge). When we reach the leaf node $x_{i, j}$, we replace the index placeholder with the corresponding value recorded in context information (e.g., we get $x_{1,2}$), and attach the actual value of all global index placeholders to the term (e.g., attach $i = 1$ to $x_{1,2}$, represented as $(1, x_{1,2})$). When we finish the traverse process, we will get $(1, x_{1,2}),\allowbreak (2, x_{2,4}), \allowbreak(1, x_{1,3}), \allowbreak(3, x_{3,4}), \allowbreak(2, -x_{1,2}), \allowbreak(4, -x_{2,4}), \allowbreak(3, -x_{1,3}), \allowbreak(4, -x_{3,4})$. By aggregating terms with the same global index, we will get the same result as \autoref{alg:exhaustive_model_instantiation}.

\begin{algorithm}
    \footnotesize
	\caption{Efficient model instantiation algorithm}
	\label{alg:efficient_model_instantiation}
	\begin{algorithmic}[1]
		\Require Symbolic multidimensional expression $ expr_G $, set data $ S^* $
		\Ensure Constraint matrix $ A $
		\Procedure{Replace}{node, context}  // node is a leaf node (term)
		    \State Replace the index placeholders of coefficient and variables with concrete values in context, and return the replaced term
		\EndProcedure
		\Procedure{Replace}{$G$, context} // G is a tuple of placeholders
		    \State Replace the global index placeholders in $G$ with concrete values in context, and return
		\EndProcedure
		\Procedure{Iterate}{node, context}
		    \If{node is a leaf node}
		        \State terms $\leftarrow$ (\Call{Replace}{node, context}, \Call{Replace}{$G$, context})
		    \ElsIf{node is the $i$th sum operator}
		        \State terms $\leftarrow$ emply list
		        \State Retrieve the concrete data $S^*_i$ of the $i$th sum operator from $S^*$ 
		        \State Filter all $(g^*_1, \cdots, g^*_N) \in S^*_i$ with the condition that $g^*_j = \text{context}(g_j)$ if $g_j$ appears in the mapping key of the context
		        \For{$ (G^*_i, L^*_i) \in S^*_i $}
		            \State context' $\leftarrow$ \Call{AddMapping}{context, $G_i \rightarrow G^*_i$, $L_i \rightarrow L^*_i$}
		            \State terms $\leftarrow$ terms $\|$ \Call{Iterate}{child, context'}
		        \EndFor
		    \ElsIf{node is an add operator}
		        \State terms $\leftarrow$ \Call{Iterate}{left, context} $\|$ \Call{Iterate}{right, context}
		    \ElsIf{node is a sub operator}
		        \State terms $\leftarrow$ \Call{Iterate}{left, context} $\|$ $-$\Call{Iterate}{right, context}
		    \EndIf
		    \State \textbf{return} terms
		\EndProcedure
		\State Initialize $A$ as an empty matrix sized $|constraints|\times|variables|$. 
		\State Initialize context as an empty mapping from symbolic index placeholder to concrete value.
		\State terms $\leftarrow$ \Call{Iterate}{root node of $expr_G$, context}
		\For{(term, $G^*) \in $terms}
		    \State Map $G^*$ and variable to matrix index $row$ and $col$
		    \State $A[row, col] = \text{coefficient of variable}$
		\EndFor
		\State \textbf{return} $ A $
	\end{algorithmic}
\end{algorithm}

\begin{proposition}\label{thm:prop1}
	The outputs of \autoref{alg:exhaustive_model_instantiation} and \autoref{alg:efficient_model_instantiation} are identical given the same input.
\end{proposition}

\begin{proof}
    To simplify the demonstration we omit the coefficient in all terms, which can be treated similarly to the variables.
    
    $\Rightarrow$: In \autoref{alg:exhaustive_model_instantiation}, assume there is a variable $x_{G^*_x, L^*_x}$ in expression $expr_{G^*}$, from \autoref{thm:lemma2} we know that the union of all $G_i$ in $I(x_{G^*_x, L^*_x})$ equals to $G$. Without loss of generality we let $I(x_{G^*_x, L^*_x}) = 1, \cdots, M$, then follow the depth-first iteration of \autoref{alg:exhaustive_model_instantiation}, we can find a sequence $(G^*_1 \| L^*_1 \in S^*_1, \cdots, G^*_M \| L^*_M \in S^*_M)$ so that $\bigcup_{i=1}^M G_i = G$ and $\bigcup_{i=1}^M G^*_i = G^*$. For \autoref{alg:efficient_model_instantiation}, we can follow the same sequence in the depth-first iteration and accumulate the mapping $G_i \rightarrow G^*_i$, $L_i \rightarrow L^*_i$ in context. Therefore when the leaf node is finally reached, the context will contain $G \rightarrow G^*$, then by the \textsc{Replace} procedure we get variable $x_{G^*_x, L^*_x}$ in expression $expr_{G^*}$ in \autoref{alg:efficient_model_instantiation}.
    
    $\Leftarrow$: In \autoref{alg:efficient_model_instantiation} assume there is a variable $x_{G^*_x, L^*_x}$ in expression $expr_{G^*}$, the context information stores mapping $G \rightarrow G^*$ when \textsc{Iterate} reached the leaf node, then following \autoref{thm:lemma2} we also have a sequence $(G^*_1 \| L^*_1 \in S^*_1, \cdots, G^*_M \| L^*_M \in S^*_M)$ so that $\bigcup_{i=1}^M G_i = G$ and $\bigcup_{i=1}^M G^*_i = G^*$. Thus when $expr_G$ is fixed into $expr_{G^*}$ in \autoref{alg:exhaustive_model_instantiation}, we can follow the same sequence and get variable $x_{G^*_x, L^*_x}$ in expression $expr_{G^*}$.
\end{proof}

\subsubsection{Parallelization of the model instantiation algorithm} \label{sec:parallelization}

While \autoref{alg:exhaustive_model_instantiation} is extremely exhaustive, it is easy to be paralleled by simply letting each worker instantiate a partition of concrete index set $\text{space}(G)$. In this section we show that \autoref{alg:efficient_model_instantiation} can also be fully paralleled by transiting index partition to data partition.

Given a set data $S^*$, its corresponding index placeholders $G$ and a partition of all possible values of $k$th index placeholder $g_k \in G$ (denoted as $\{P_1, \cdots, P_M\}, \bigcup_{i}P_i = \text{space}(g_i)$), a data partition of set data $S^*$ over $k$th index $g_k$ is to partition $S^*$ to $M$ subset $S^*_1, \cdots, S^*_M$, so that $S^*_i = \{(g^*_1, \cdots, g^*_N) | \allowbreak (g^*_1, \cdots, g^*_N) \in S^*, g^*_k \in P_i\}$. That is, the $i$th subset of the data only contains the index tuple whose $k$th element is in $i$th partitioned index subset $P_i$. An example is shown in \autoref{fig:data_partition}.

\begin{figure}
	\centering
	\includegraphics[width=0.5\linewidth]{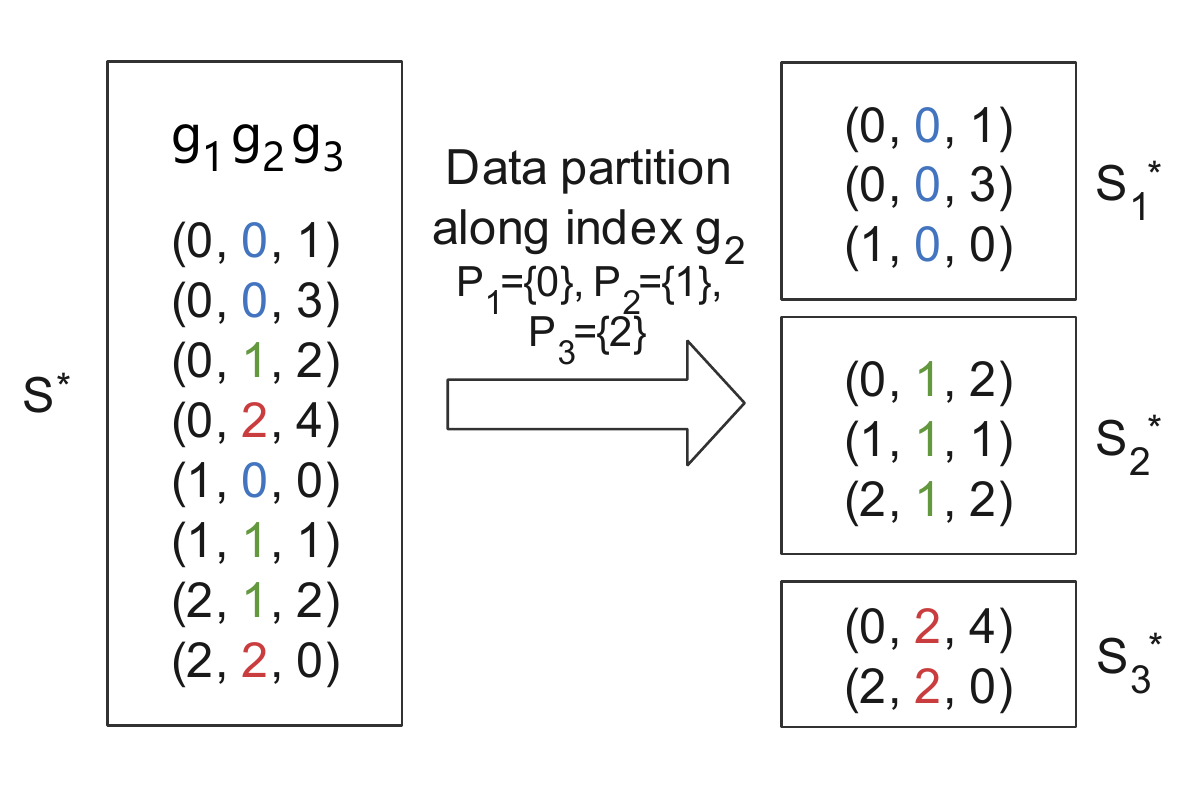}
	\vspace{-1em}
	\caption{An example of data partition.}
	\vspace{-1em}
	\label{fig:data_partition}
\end{figure}

Then we have the following proposition:

\begin{proposition}\label{thm:prop2}
	Running \autoref{alg:exhaustive_model_instantiation} on a subset of all possible global index value $(g^*_1, \cdots, g^*_N) \in \text{space}(G)$ with condition $g^*_k \in P$, is equivalent to running \autoref{alg:efficient_model_instantiation} on a subset of data $(g^*_1, \cdots, g^*_N) \in S^*$ with condition $g^*_k \in P$.
\end{proposition}

\begin{proof}
    The proof is similar to \autoref{thm:prop1} with some details on sequence sharing between algorithms. For $\Rightarrow$, the sequence $(G^*_1 \| L^*_1 \in S^*_1, \cdots, G^*_M \| L^*_M \in S^*_M)$ can still be transferred to \autoref{alg:efficient_model_instantiation} whose data $S^*$ is filtered by condition $g^*_k \in P$, since $\text{space}(G)$ is also filtered in \autoref{alg:exhaustive_model_instantiation} to make sure only $g^*_k \in P$ will appear in the sequence. For similar reason $\Leftarrow$ holds.
\end{proof}

With \autoref{thm:prop2} we can fully parallelize \autoref{alg:efficient_model_instantiation}. For example, to distribute the instantiation process of \eqref{eq:expr_i} equally to two workers
, by applying \autoref{thm:prop2}, we can do data partition on $E^*$ as $E^*_1 = \{(1, 2), (2, 4), (1, 3)\}, E^*_2 = \{(1, 2)\}$ for worker 1, $E^*_1 = \{(3, 4)\}, E^*_2 = \{(2, 4), (1, 3), (3, 4)\}$ for worker 2. Applying \autoref{alg:efficient_model_instantiation}, worker 1 will generate $(1, x_{1,2}),\allowbreak (2, x_{2,4}), \allowbreak(1, x_{1,3}), \allowbreak(2, -x_{1,2})$, worker 2 will generate $(3, x_{3,4}), \allowbreak(4, -x_{2,4}), \allowbreak(3, -x_{1,3}), \allowbreak(4, -x_{3,4})$. It is easy to check that the results are identical.

\section{Sequential Decomposition of Large-scale Optimization}\label{sec:decomposition}

The second challenge for a rapid AMS is the solving time. Extremely large-scale mathematical optimization models usually take huge amount of time to solve, and the scale of data may even bump up due to sudden business need (e.g., big promotion), resulting in potential risk of timeout or even unsolvability. In this section, we introduce a lossy decomposition method (Guided FRH) for massive sequential models. While the feasibility of solutions is strictly maintained, the decomposition methods make a trade-off between optimality and efficiency.

\subsection{Preliminaries: Sequential Decision Making and Rolling Horizon} \label{sec:RH}

Sequential (or dynamic) decision making widely exists in applied scenarios. For example, we may need to decide the production amount of specific items in a range of dates, in which the prior decisions will influence successive ones. More formally, for a sequential model of linear constraints with $T$ periods, its decision variables $\textbf{x}$ can be divided into $T$ row vectors $\textbf{x}_1, \cdots, \textbf{x}_T$, so that the constraints can be formulated as
\begin{equation*}
    \begin{aligned}
        &\textbf{A}_1 \textbf{x}_1^T &= \textbf{b}_1 \\
        &\textbf{A}_2 [\textbf{x}_1, \textbf{x}_2]^T &= \textbf{b}_2 \\
        &\cdots &\\
        &\textbf{A}_T [\textbf{x}_1, \textbf{x}_2, \cdots, \textbf{x}_T]^T &= \textbf{b}_T
    \end{aligned}
\end{equation*}
which indicates that the constraint matrix has a block triangular structure. Here we assume the decision variables share the same semantic meaning in each period, and use $x^i_t$ to denote the $i$th variable in period $t$. We refer to \cite{BradleyStephenP.1977Amp} for a detailed introduction.

To make sequential decisions in a dynamic environment, a common business practice is rolling horizon (RH)\cite{Sethi1991}. That is, we do planning in a relatively long time window (planning horizon) using the latest available information, and only accept the generated decisions in the first several time steps (re-planning periodicity). When the re-planning periodicity passed, we start a new planning horizon with updated environment information, and repeat the above procedure. 

From this perspective, compared with the decision in the re-planning periodicity that will be actually applied, the planning after the re-planning periodicity is more likely a ``simulation-based guidance'' to reach the global optimum. That is, although the decisions after the re-planning periodicity is never actually executed, we assume that they will be executed on a simulation basis, so that the global optimum in a longer range is considered. 

In our large-scale scenario, the size of the planning horizon is limited due to the scalability of solver engines. To enlarge the range of future information involved in the optimization and reach global optimum in a longer time window, we adopt the idea of ``guidance'' in rolling horizon, but with a more computational efficient approach.

\subsection{Forward Rolling Horizon} \label{sec:FRH}

To begin with, we introduce a simple method, Forward Rolling Horizon (FRH) \cite{DIMITRIADIS1997S1061} to illustrate the basic idea. For a large sequential model $ P $ with length $ T $, we divide the model into $ h $ sub-models $ P_1, \cdots, P_h $. Each sub-model $ P_i $ starts at sequence period $ t_i $ and ends at period $ t_{i+1} - 1 $. The sub-problems are solved in a sequential manner so that the model $ P_i $ can take advantage of the solutions of $ P_1, \cdots, P_{i-1} $. To guide each sub-model $ P_i $ towards the global optimum, we aggregate the future information from period $ t_{i+1} $ to the last period into $M$ periods ($M = 1$ by default), and attach it to the end of the sub-model. Therefore, except for the last sub-model, each sub-model consists of $ t_{i+1} - t_i + M $ periods. The FRH optimization procedure is shown in \autoref{fig:FRH}.

\begin{figure}
	\centering
	\begin{subfigure}{\linewidth}
	    \centering
		\includegraphics[width=0.8\linewidth]{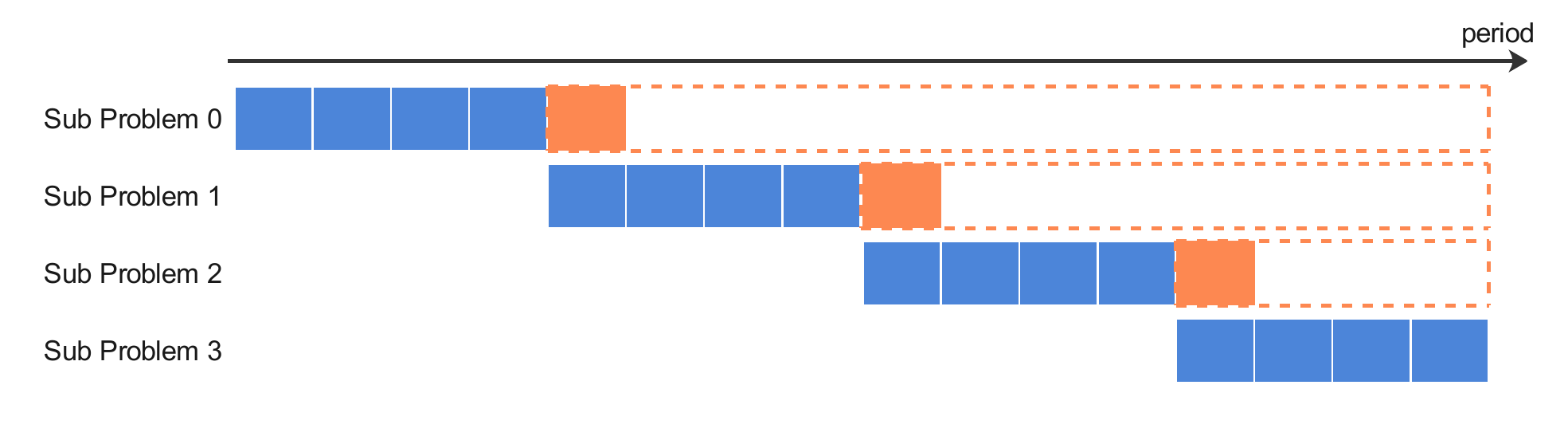}
	    \caption{Forward RH (FRH), \cite{DIMITRIADIS1997S1061}}
		\label{fig:FRH}
	\end{subfigure}
	\begin{subfigure}{\linewidth}
	    \centering
		\includegraphics[width=0.8\linewidth]{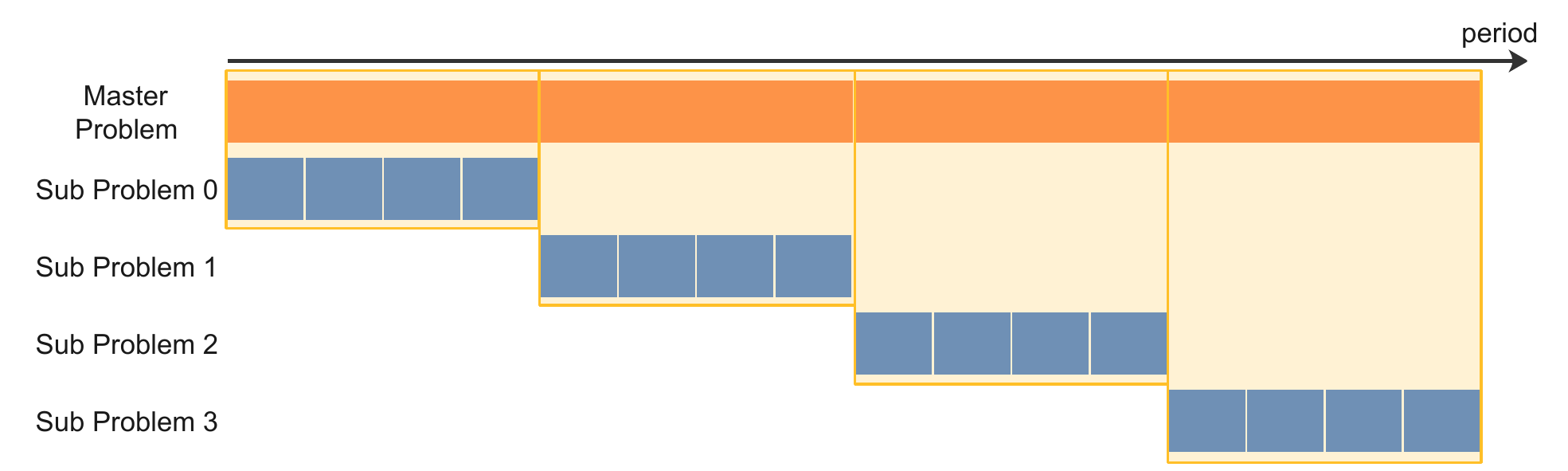}
		\caption{Guided RH}
		\label{fig:Guided_RH}
	\end{subfigure}	
	\begin{subfigure}{\linewidth}
	    \centering
		\includegraphics[width=0.8\linewidth]{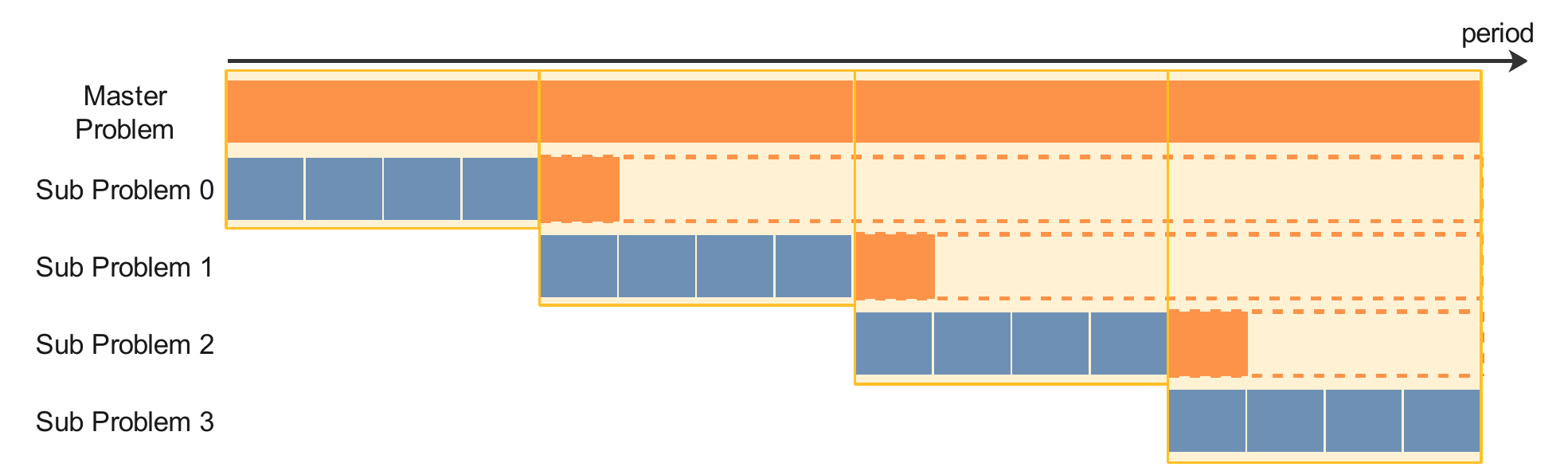}
		\caption{Guided FRH}
		\label{fig:Guided_FRH}
	\end{subfigure}	
	\vspace{-1em}
	\caption{Forward rolling horizon and proposed decomposition methods for large-scale mathematical optimization.}
	\vspace{-1em}
\end{figure}

\subsection{Guided Rolling Horizon} \label{sec:GRH}

In this section, we propose Guided Rolling Horizon (Guided RH), which decompose a raw sequential problem into a main problem and a sequence of sub-problems, with all constraints strictly satisfied and global optimality largely preserved.

For a raw sequential problem with $ T $ periods and $N$ decision variables in each period, the Guided RH optimization plan consists of three steps:

\begin{description}	
	\item[Data aggregation] Aggregate all feed-in data from $ T $ periods to $ h $ periods, in which $ h $ is the number of sub-models. Every data element at period $ i $ in the aggregated data represents the data elements from period $ t_i $ to $ t_{i+1} - 1 $ in the original data. E.g., aggregate data from daily to weekly.
	\item[Master problem solving] Construct a ``shrunk'' master problem with $ h $ periods using the aggregated data, then solve it. The master problem provides a high-level, soft guidance to the actual rolling-horizon procedure afterwards, therefore all the integer constraints can be omitted. Similar to data aggregation step, every solution element $ z^i_k $ at period $ k $ in this master problem represents the solution elements $ x^i_{t_k}, \cdots, x^i_{t_{k+1} - 1} $ in the original problem.
	\item[Sub problems solving] To transit the master problem's aggregated solutions to an applicable solutions, we apply the rolling horizon method on the original problem, with additional guided constraints and objectives to make the two solutions aligned. In each horizon $ k $ (from period $ t_k $ to period $ t_{k+1} - 1 $), we add soft constraints to each variable group along period axis, to minimize the gap between applicable solutions $ x^i_{t_k}, \cdots, x^i_{t_{k+1} - 1} $ and the aggregated solution $ z^i_k $. i.e., 
	\begin{equation*}\label{eq:guidance}
	\begin{aligned}
	    \min \quad &obj + \sum_{i=1}^N \lambda_k (u^i_k + v^i_k) \\
		s.t. \quad &\textbf{A}_k [\textbf{x}_1, \textbf{x}_2, \cdots, \textbf{x}_k]^T = \textbf{b}_k \\
		&\sum_{j=t_k}^{t_{k+1} - 1} x^i_j - z^i_k = u^i_k - v^i_k, \quad \forall i \in \{1, \cdots, N\} \\
		&u^i_k \geq 0, v^i_k \geq 0, \quad \forall i \in \{1, \cdots, N\}
	\end{aligned}
	\end{equation*}
	in which $\textbf{x}_1, \cdots, \textbf{x}_{k-1}$ are already solved and fixed in previous sub-problem solving. $obj$ is the original objective of the problem, $ u_i $ and $ v_i $ are auxiliary variables that help minimize the L1 loss $ \sum_i |\sum_{j=t_k}^{t_{k+1} - 1} x^i_j - z^i_k| $. $ \lambda_k $ is the weight controlling to what extent should the applicable solution be aligned with the aggregated solution in horizon $ k $. It is optional to re-solve the master problem with current horizon's solution fixed to reduce the cumulative error between master and sub-problems. It is also possible for sub-problems to overlap with each other. For many real scenarios, only the solution of first $ K $ periods is needed, in this case we can stop after solving $ l $ horizons so that $ t_{l+1} > K $.
\end{description}

The procedure of Guided RH is shown in \autoref{fig:Guided_RH}. 

\subsection{Guided FRH} \label{sec:GFRH}

While both FRH and Guided RH guide rolling horizon towards long-term global optimality, there is no conflict between them. Therefore, we can further improve the optimality by replacing the rolling horizon solving procedure of Guided RH with FRH so as to combine the two methods. The procedure of Guided FRH is shown in \autoref{fig:Guided_FRH}.

\subsection{Fine-tuning of approximated solutions}\label{sec:fine-tuning}

\begin{figure}
	\centering
	\includegraphics[width=0.8\linewidth]{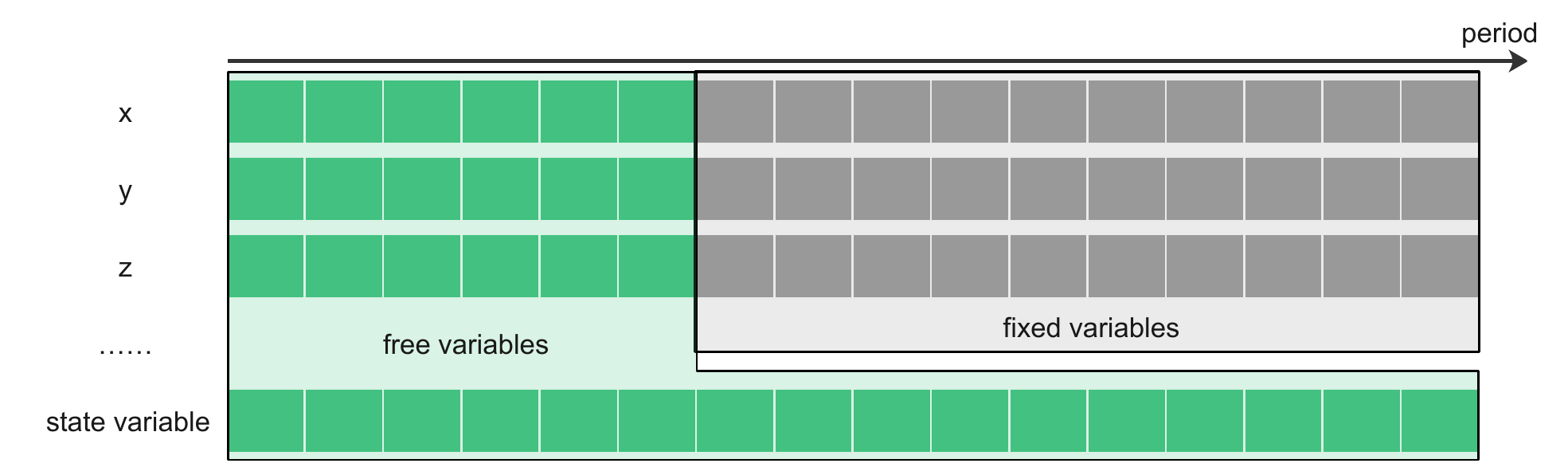}
	\vspace{-1em}
	\caption{The fine-tuning procedure. The grey variables are fixed while the green variables are to be re-optimized. State variables whose value are determined by other variables keep free in the whole sequence.}
	\vspace{-1em}
	\label{fig:fine_tuning}
\end{figure}

When we get an approximated solution $ \textbf{x}_1, \cdots, \textbf{x}_T $ of a sequential model with the above decomposition methods, it will be beneficial if we can do some fine-tuning to the solution to further boost the global optimality. A possible method is to re-optimize the variable in the first $ k $th periods, with other variables fixed. Since the solution is already a feasible one, the re-optimized solution will not be worse than the original solution. Notice that there might be some ``state variables'' whose value are fully determined by other variables. These variables will keep free in the whole sequence, so that the change of previous decisions can still propagate to the later periods. The procedure of fine-tuning is shown in \autoref{fig:fine_tuning}.

\section{Evaluation}

In this section, we first conduct offline benchmarks against current state-of-the-art AMSs on model instantiation efficiency, and then deployed Grassland in Huawei's production planning scenario for nearly half a year. By providing near-real-time production planning simulation, Grassland plays an essential role in Huawei's supply chain management against highly dynamic supply-demand environment.

\begin{table}
	\centering
	\small
	\begin{tabular}{c|rrr}
		\hline 
		& {P-Median} & {\makecell{Offshore\\Wind Farming}} & {\makecell{Food\\Manufacture I}} \\ 
		\hline 
		Gurobi Py API	            & 410.20 & 533.71 & 744.39 \\ 
		JuMP 						& 278.08 & 169.08 & 789.86 \\ 
		ZIMPL 						& 174.00 & 400.47 & 399.16 \\
		AMPL                        & 15.94  & 17.71 & 31.65 \\
		Grassland (S) & 35.91  & 18.85 & 80.83 \\ 
		Grassland (M) 	& \textbf{2.09} & \textbf{1.67} & \textbf{5.28} \\ 
		\hline 
	\end{tabular} 
	\caption{Model Instantiation Benchmark. Total time (in seconds) to process the model definition and produce the output file in CPLEX LP format.}
	\vspace{-1em}
	\label{tab:evaluation}
\end{table}

\begin{figure*}[htbp]
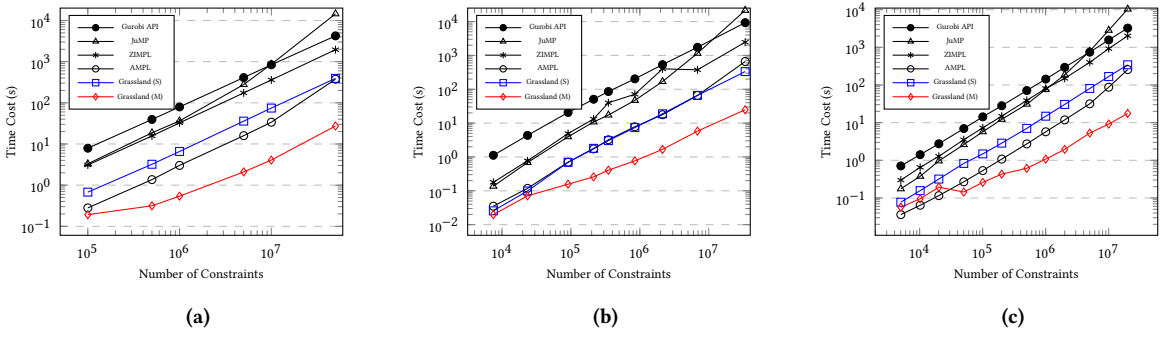

	\centering
	\begin{subfigure}[b]{0.3\linewidth}
        \include{pgfplots/pm}
        \vspace{-10pt}
		\caption{}
		\label{fig:evaluation_PM}
	\end{subfigure}
	\begin{subfigure}[b]{0.3\linewidth}
        \include{pgfplots/owf}
        \vspace{-10pt}
		\caption{}
		\label{fig:evaluation_OWF}
	\end{subfigure}	
	\begin{subfigure}[b]{0.3\linewidth}
		\include{pgfplots/fmi}
		\vspace{-10pt}
		\caption{}
		\label{fig:evaluation_FMI}
	\end{subfigure}	
	\vspace{-1em}
	\caption{Offline model instantiation benchmark on (a) P-Median (b) Offshore Wind Farming (c) Food Manufacture I.}
	\vspace{-0.5em}
	\label{fig:evaluation}
\end{figure*}

\begin{figure*}[htbp]
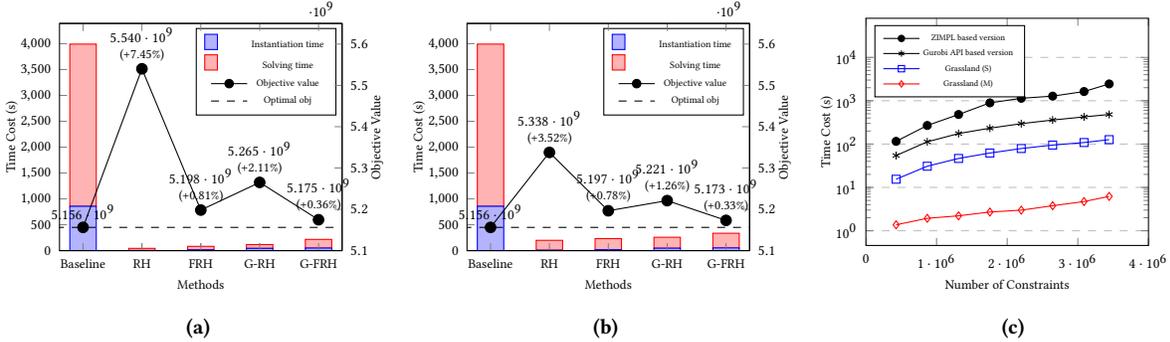

	\centering
	\begin{subfigure}[b]{0.3\linewidth}
        \include{pgfplots/decomposition}
        \vspace{-10pt}
		\caption{}
		\label{fig:decomposition_orig}
	\end{subfigure}
	\begin{subfigure}[b]{0.3\linewidth}
        \include{pgfplots/decomposition_fine_tuning}
        \vspace{-10pt}
		\caption{}
		\label{fig:decomposition_fine_tuning}
	\end{subfigure}	
	\begin{subfigure}[b]{0.3\linewidth}
        \include{pgfplots/sd}
        \vspace{-10pt}
		\caption{}
		\label{fig:evaluation_SD}
	\end{subfigure}	
	\vspace{-1em}
	\caption{Online experiment results. (a) Comparison between baseline and different decomposition methods on time cost and optimality. (b) Same as (a) with fine-tuning of first 20 periods. (c) Comparison of model instantiation efficiency.}
	\vspace{-0.5em}
	\label{fig:decomposition}
\end{figure*}

\subsection{Offline Model Instantiation Benchmark}

Before online experiment, we first benchmark our system against several widely-used modeling software on some typical mathematical optimization problems, to evaluate the scalability and effectiveness of proposed model instantiation method.

\subsubsection{Test Problems}

We select three typical problems from JuMP \cite{LubinDunningIJOC} and Gurobi modeling examples \cite{gurobi_modeling_examples}.

\begin{itemize}
	\item \textbf{P-Median}: This model is used in \cite{hart2011pyomo,LubinDunningIJOC} to compare the modeling efficiency of Pyomo and JuMP with other AMSs.
	\item \textbf{Offshore Wind Farming}: This model in Gurobi modeling example minimize the cost of laying underwater cables to collect electricity produced by an offshore wind farm.
	\item \textbf{Food Manufacture I}: This model in Gurobi modeling example is a blending problem. Multiple raw materials are combined in a way to achieve the lowest cost.
\end{itemize}

\subsubsection{Compared Modeling Softwares}

In the following experiments, we compare the following modeling software. Modeling software implemented in interpreted languages such as 
Pyomo \cite{hart2011pyomo} are not included due to lack of efficiency in the prior benchmark works \cite{hart2011pyomo,LubinDunningIJOC}.

\begin{itemize}
	\item \textbf{Gurobi Modeling API} \cite{gurobi2018gurobi}: The Python interface of Gurobi, which provides the standard implementation of Gurobi modeling examples.
	\item \textbf{ZIMPL} \cite{Koch2004}: The modeling language in SCIP Optimization Suite, written in plain C.
	\item \textbf{JuMP} \cite{LubinDunningIJOC}: A modern AML implemented in Julia, which is reported to achieve high efficiency comparable to commercial products such as AMPL, with the benefits of remaining in a high-level modern language. 
	\item \textbf{AMPL} \cite{fourer1990modeling}: A leading, widely-used commercial AMS.
	\item \textbf{Grassland (Single-threaded)}: The implementation of our proposed method in \autoref{sec:model_instantiation} with only one thread.
	\item \textbf{Grassland (Multi-threaded)}: The implementation of our proposed method, including parallelization described in \autoref{sec:parallelization} (64 threads).
\end{itemize}

\subsubsection{Results}

The benchmark result on all the problems is shown in \autoref{tab:evaluation}. It is shown that Grassland achieve absolute predominance over all other modeling software. While single-threaded Grassland already achieve 4-5x speedup over the fastest open-source modeling software (ZIMPL) and comparable with the leading commercial software (AMPL), multi-threaded Grassland further achieves 6-10x speedup over AMPL.

We also tested the software on different scales of models. The result is shown in \autoref{fig:evaluation}. It is shown that Grassland has superior performance over all scale of models.

\subsection{Online Experiment}

\subsubsection{Background: Production planning and supply-demand analysis}

Production planning is the planning of production activities to transform raw materials (supply) into finished products, meeting customer's order (demand) in the most efficient or economical way possible. It lies in the core of manufacturing companies' supply chain management, directly influencing the profit and customer satisfaction. Mathematical optimization is a mainstream method for production planning, and \cite{10.5555/1951836} provides a comprehensive introduction to this. As a world-leading electronic manufacturer, Huawei provides more than ten thousand kinds of end products, with even much more kinds of raw materials and intermediate assemblies, to satisfy millions of demands from all over the world. Modeling in such a large-scale scenario involves millions of variables and constraints. 

While mathematical optimization can deliver near-optimal production planning solution for a certain input of supply and demand, The demand and supply itself is always changing with high uncertainty due to external dynamic factors
. Our planning department needs to react quickly to such changes to ensure business continuity, which is called \textit{supply-demand analysis}. 
To support the analysis, a crucial process is \textit{production planning simulation} which returns the final planning result (e.g., fulfillment rate) for a certain input of supply and demand, helping planners evaluate and improve their analysis. For example, when several raw materials are suddenly unavailable, planners may try different ways to increase the supply of alternative materials and run production planning simulation for each of them, adjust the supply iteratively to increase the fulfillment rate of related end products, and get the final supply adjustment decision.

\subsubsection{Dataset and Compared Methods}

The dataset is from real production environment which consists of all needed manufacturing data in 78 weeks (one week per period). The instantiated model consists of 5,957,634 variables, 3,443,465 constraints and 30,366,971 nonzero elements in the constraint matrix. In the following experiments, we compare the result before and after the application of Grassland in the production planning simulation scenario. The baseline is the original method before the deployment of Grassland, which is quite standard: ZIMPL is used for modeling and instantiation, and the instantiated full model is directly solved without any decomposition. The decomposition methods in \autoref{sec:decomposition}  (RH, FRH, Guided RH, Guided FRH) are tested separately. Sequence length $T = 78$, number of submodels $h = 8$, fine-tuning periods $k = 20$. All models are solved via Mosek ApS\cite{mosek}. \footnote{Note that we cannot deploy more efficient solvers like CPLEX on Huawei's online enterprise environment due to export restriction of the US (so as advanced AMSs like AMPL). However, all solvers always deliver exact optimal solution if possible, and have similar exponential curves between problem scale and solving time. Therefore the selection of external solvers will not change the optimality and solving time ratio.}


\subsubsection{Results}

The main result is shown in \autoref{fig:decomposition_orig}. Our proposed methods can achieve significant acceleration (15-35x faster) than baseline method, while the feasibility is strictly maintained and the loss of objective is small. For Guided FRH, it can achieve 15x acceleration with only 0.36\% of the objective loss. Practically, such a tiny optimality loss does not cause sensible issues ($0.1\% \sim 0.2\%$ fluctuation of fulfillment ratio), especially considering that multiple source of more dominant error exist in complex business models such as prediction and approximation error.

The fine-tuning result is shown in \autoref{fig:decomposition_fine_tuning} in which the first 20 periods of the problem is re-optimized following \autoref{sec:fine-tuning}. It is shown that fine-tuning can significantly narrow the gap between the objective value of decomposition methods and the optimal one.

Additionally, for instantiation efficiency in online scenarios, we compare Grassland with two legacy systems that we previously developed and deployed in online environment, based on ZIMPL and Gurobi Modeling API respectively. We use the number of periods involved to control the model size. The result is shown in \autoref{fig:evaluation_SD}. It is shown that the result is aligned with offline benchmarks in \autoref{fig:evaluation}.


\section{Conclusion}

In this paper, we propose Grassland, an algebraic modeling system that is efficient in large-scale mathematical optimization scenarios, including a parallelized instantiation scheme for general linear constraints, and a lossy sequential decomposition method that accelerates large-scale model solving exponentially. We perform both offline benchmarks and online deployment in Huawei's production planning scenario. The results demonstrate the significant superiority of Grassland over strong baselines.


\section*{Appendix}
\appendix


\section{Implementation Details}

As an algebraic modeling system, Grassland is implemented in five layers:

\begin{description}	
	\item[Modeling API (AML)] The grassland modeling API is implemented as a Python package \texttt{grassland} (\texttt{gl}). 
    \item[Intermediate representation (IR) layer] This layer plays as a bridge between Modeling API and highly-efficient c++ backend. Models defined by Grassland modeling API is translated into a unified, JSON-based intermediate representation with six components (variables, constants, index placeholders, expression graphs, constraints and bounds).
    \item[Decomposition layer] Implements four decomposition methods in \autoref{sec:decomposition} (RH, FRH, Guided RH, Guided FRH).
    \item[Model instantiation layer] Implements the parallelized model instantiation scheme in \autoref{sec:instantiation_alg} and \autoref{sec:parallelization}. The parallelization is implemented by multi-threaded programming. Due to the extreme efficiency of our proposed method, even the float-to-string conversion becomes a significant bottleneck. Here we apply Ryu \cite{ryu} to accelerate the conversion.
    \item[Solver layer] Calls different solver engine to solve the instantiated model or sub-model and return back the solution.
\end{description}

\section{Integration with rounding procedure}

In applied optimization pipeline, there usually exist some integer constraints for decision variables. To achieve this, a mixed integer programming model or an external heuristics rounding procedure may apply. However, the efficiency of rounding procedure can heavily rely on the scale of the model, as well as the number of integer constraints. With the above decomposition methods, we can round the variables at the same time when we solve each sub-model. Since the sub-model is significantly smaller than the original one, the rounding procedure will also be largely accelerated.

\section{Test problems for Model Instantiation}\label{sec:test_problems}

\begin{table}
	\centering
	\small
	\begin{tabular}{c|rrr}
		\hline 
		Test Problem & \#(variables) & \#(constraints) & \#(nonzeros) \\
		\hline 
		P-Median & 5,050,000 & 5,000,164 & 15,050,000 \\ 
		Offshore Wind Farming & 4,170,120 & 4,220,120 & 10,425,300 \\ 
		Food Manufacture I & 5,006,394 & 14,974,502 & 39,932,750 \\
		\hline 
	\end{tabular} 
	\caption{The basic statistics of the benchmark test problems in \autoref{tab:evaluation}.}
	\label{tab:test_info}
\end{table}

The basic statistics of the benchmark test problems in \autoref{tab:evaluation} is listed in \autoref{tab:test_info}. Problem data is randomly generated while maintaining the feasibility, and we control the size of the data to generate different scale of models in \autoref{fig:evaluation}.

\section{Expreimental Setting}

For offline model instantiation benchmark, all benchmarks are run on a server with 32-core (64-thread) CPUs and 192GB memory. The output format of all constructed problems is set to CPLEX LP format (a standard LP/MIP format that is supported by most of the mathematical solvers). The identity of all constructed problems by different softwares on small and medium size are checked by actual solving with Gurobi with the same optimized objective value, and checked by static comparison script for extremely large size that cannot be directly solved in reasonable time. MIP problems are relaxed into LP in the identity checking process to extend scalability. For multi-threaded Grassland, the size of thread pool is set to 64. The basic information of the benchmark problems are shown in \autoref{tab:test_info}.

For online experiment. all experiments are run on a Huawei cloud server with 32-core (64-thread) CPUs and 256GB memory. All models are solved via Mosek ApS\cite{mosek}. For production planning problem, the sequence length $ T = 78 $ and we use number of sub-models $ h = 8 $, fine-tuning periods $ k = 20 $ for all decomposition methods.

Software version: 

\begin{itemize}
	\item Gurobi Modeling API (7.5.1, released in Jul 2017)\footnote{Due to export restrictions, we cannot purchase and deploy the latest version of Gurobi in Huawei's enterprise environment.}
	\item ZIMPL (3.4.0, released in June 2020)
	\item JuMP (0.21.3, released in June 2020)
	\item AMPL (20200810, released in Aug 2020)
\end{itemize}

\section{Experimental Result}

\begin{table}
	\centering
	\small
	\begin{tabular}{@{\hspace{3pt}}c@{\hspace{3pt}}|@{\hspace{3pt}}r@{\hspace{6pt}}r@{\hspace{6pt}}r@{\hspace{6pt}}r@{\hspace{3pt}}}
		\hline 
		Method & {\makecell{Model\\instantiation time}} & {\makecell{Solving\\time}} & {\makecell{Total\\time}} & {\makecell{Objective\\($\times 10^9$)}} \\
		\hline 
		Baseline & 857.00 & 3135.16 & 3992.16 & 5.15626 \\ 
		\hline 
		RH & 4.94 & 40.80 & 45.74 & (+7.45\%) 5.54017 \\ 
		FARH & $ 8.54^* $ + 6.78 & 67.68 & 83.00 & (+0.81\%) 5.19818 \\ 
		G-RH & $ 28.84^{**} $ + $ 8.79^* $ + 6.20 & 73.31 & 117.14 & (+2.11\%) 5.26513 \\ 
		G-FARH & $ 30.99^{**}$ + $11.07^*$ + 7.43 & 167.62 & 217.11 & \textbf{(+0.36\%) 5.17477} \\ 
		\hline 
	\end{tabular} 
	\caption{Sequential decomposition benchmark on demand-supply analysis problem. In ``Model instantiation time'' column, time marked with ``*'' is the time for data compression, time marked with ``**'' is the time to generate guided constraints and objectives.}
	\label{tab:decomposition}
\end{table}

\begin{table}
	\centering
	\small
	\begin{tabular}{@{\hspace{0pt}}c@{\hspace{1pt}}|@{\hspace{1pt}}r@{\hspace{3pt}}r@{\hspace{3pt}}r@{\hspace{3pt}}r@{\hspace{0pt}}}
		\hline 
		Method & {\makecell{Model\\instantiation time}} & {\makecell{Solving\\time}} & {\makecell{Total\\time}} & {\makecell{Objective\\($\times 10^9$)}} \\
		\hline  
		Baseline & 857.00 & 3135.16 & 3992.16 & 5.15626 \\ 
		\hline 
		RH & 8.16 & $ 150.35^\dagger $ + 40.80 & 199.31 & (+3.52\%) 5.33775 \\ 
		FARH & $ 8.54^* $ + 9.19 & $ 147.11^\dagger $ + 67.68 & 232.52 & (+0.78\%) 5.19651 \\ 
		G-RH & $ 28.84^{**} $ + $ 8.79^* $ + 9.29 & $ 138.11^\dagger $ + 73.31 & 258.34 & (+1.26\%) 5.22104 \\ 
		G-FARH & $ 30.99^{**} $ + $ 11.07^* $ + 9.93 & $ 116.41^\dagger $ + 167.62 & 336.02 & \textbf{(+0.33\%) 5.17321} \\ 
		\hline 
	\end{tabular} 
	\caption{Fine-tuning benchmark on demand-supply analysis problem. In ``Solving time'' column, time marked with ``$ \dagger $'' is the extra solving time for fine tuning.}
	\label{tab:fine-tuning}
\end{table}

The detailed result of online experiment is shown in \autoref{tab:decomposition} and \autoref{tab:fine-tuning}.

\section{Production Planning Model}

While real-world production planning models are complex with lots of variants for different scenarios, here we show a self-contained, simplified version with only three types of constraints. We refer to \cite{10.5555/1951836} for a detailed introduction.

\allowdisplaybreaks
\begin{align}
    \min \quad &\sum_{t,p,i}C^m_{t,p,i}m_{t,p,i} + \sum_{t,p,i}C^x_{t,p,i}x_{t,p,i} + \sum_{t,p,i}C^{pur}_{t,p,i}pur_{t,p,i} \nonumber\\
    & + \sum_{t,p,i,i',j}C^{rp}_{t,p,i,i',j}rp_{t,p,i,i',j} + \sum_{t,p,i,j}C^r_{t,p,i,j}r_{t,p,i,j}\nonumber\\
    s.t. \quad & inv_{t, p, i} = inv_{t-1, p, i} + inbound_{t, p, i} - outbound_{t, p, i} \label{eq:inv_constraint}\\
    & inbound_{t, p, i} = \sum_{t'} x_{t', p, i} + \sum_{t'} pur_{t', p, i} + \sum_{p'} s_{t, p', p, i} \nonumber\\
    & + \sum_{i', j} rp_{t, p, i, i', j} + \sum_{j} r_{t, p, j, i} + PO_{t, p, i} + WIP_{t, p, i} \nonumber\\
    & outbound_{t, p, i} = \sum_{j} B_{t, p, i, j} x_{t, p, j} + \sum_{p'} s_{t, p, p', i} + \sum_{i', j} rp_{t, p, i', i, j} \nonumber\\
    & + \sum_{j} r_{t, p, i, j} + z_{t, p, i} \nonumber\\
    & \text{for } (t, p, i) \in P \nonumber\\
    & m_{t, p, i} = m_{t-1, p, i} - z_{t, p, i} + D_{t, p, i} \label{eq:delay}\\
    & \text{for } (t, p, i) \in P \nonumber\\
    & \sum_{i'} rp_{t, p, i, i', j} \leq B_{t, p, i, j}x_{t, p, j} \label{eq:rp_max} \\
    & \text{for } (t, p, i, j) \in BOM \nonumber
\end{align}

Indices:

\begin{itemize}
    \item $i, i', j$: items (raw material, sub-assembly or end product).
    \item $p$: plant.
    \item $t$: period.
\end{itemize}

Decision variables:

\begin{itemize}
    \item $x_{t,p,i}, pur_{t,p,i}, z_{t,p,i}$: the production/purchase/deliver amount of item $i$ in plant $p$ at period $t$.
    \item $s_{t,p,p',t}$: the transit amount of item $i$ from plant $p$ to plant $p'$ at period $t$.
    \item $r_{t,p,i,j}$: the amount that item $i$ replace item $j$ in plant $p$ at period $t$.
    \item $rp_{t,p,i,i',j}$: the amount that item $i'$ replace item $i$ to produce item $j$ in plant $p$ at period $t$.
\end{itemize}

State variables (whose value is determined by other decision variables):

\begin{itemize}
    \item $inv_{t,p,i}$: the inventory amount of item $i$ in plant $p$ at period $t$.
    \item $m_{t,p,i}$: the delay amount of item $i$ in plant $p$ at period $t$.
\end{itemize}

All the value of decision and state variables are not less than zero.

Some important constants (note that not all constants are listed due to space limit. Every sum operation contains a constant that controls the range of indices):

\begin{itemize}
    \item $C^m, C^s, C^{pur}, C^{rp}, C^{r}$: the cost of delay, transition, purchase and replacement. (Delay cost will usually dominate the objective)
    \item $P[t, p, i]$: item $i$ will be produced in plant $p$ at period $t$.
    \item $BOM[t, p, i, j]$: $j$ is the parent of $i$ in plant $p$ at period $t$.
    \item $B_{t,p,i,j}$: number of item $i$'s amount that need to be consumed to producing one item $j$.
    \item $PO_{t,p,i}, WIP_{t,p,i}$: the amount of purchase order (PO) / work-in-progress (WIP) of item $i$ in plant $p$ at period $t$.
\end{itemize}

Constraints:

\begin{itemize}
    \item \autoref{eq:inv_constraint}: inventory constraint. The current inventory amount equals to last period's inventory plus inbound minus outbound.
    \item \autoref{eq:delay}: delay constraint. The current delay amount equals to last period's delay amount plus delivery amount minus demand amount.
    \item \autoref{eq:rp_max}: replacement constraint. For all component-assembly relation, the sum of replacement amount cannot exceed the needed amount for assembly's production.
\end{itemize}

\bibliographystyle{ACM-Reference-Format}
\bibliography{references}

\end{document}

%% file: pgfplots/pm.tex
\begin{tikzpicture}
\begin{loglogaxis}[
    tiny,
    legend style={nodes={scale=0.6, transform shape}},
    width=\textwidth,
    xlabel={Number of Constraints},
    xlabel shift= -2pt,
    ylabel={Time Cost (s)},
    ylabel shift= -5pt,
    xmin=50000, xmax=60000000,
    ymin=0, ymax=20000,
    xtick={100000,1000000,10000000},
    ytick={0.1,1,10,100,1000,10000},
    legend pos=north west,
    ymajorgrids=true,
    grid style=dashed,
    mark size=1.5pt,
    scaled ticks=false
]

\addplot[color=black, mark=*]
    coordinates {
    (100101,7.8993)(500101,39.4655)(1000101,79.4978666666667)(5000101,410.197366666667)(10000101,840.069266666667)(50000101,4230.0839)
    };
    \addlegendentry{Gurobi API}
    
\addplot[color=black, mark=triangle]
    coordinates {
    (100101,3.30533333333333)(500101,18.6683333333333)(1000101,36.376)(5000101,278.079666666667)(10000101,863.207)(50000101,14518.064)
    };
    \addlegendentry{JuMP}

\addplot[color=black, mark=asterisk]
    coordinates {
    (100101,3.08656666666667)(500101,15.6893)(1000101,32.6420666666667)(5000101,174.001166666667)(10000101,359.682533333333)(50000101,1953.8551)
    };
    \addlegendentry{ZIMPL}
    
\addplot[color=black, mark=o]
    coordinates {
    (100101,0.280340333333333)(500101,1.36189666666667)(1000101,3.02159666666667)(5000101,15.9434666666667)(10000101,33.7797666666667)(50000101,374.249)
    };
    \addlegendentry{AMPL}

\addplot[color=blue, mark=square]
    coordinates {
    (100101,0.676533333333333)(500101,3.21753333333333)(1000101,6.58343333333333)(5000101,35.9115333333333)(10000101,74.5211333333333)(50000101,392.210833333333)
    };
    \addlegendentry{Grassland (S)}

\addplot[color=red, mark=diamond]
    coordinates {
    (100101,0.192333333333333)(500101,0.3152)(1000101,0.540833333333333)(5000101,2.0926)(10000101,4.06783333333333)(50000101,27.3155)
    };
    \addlegendentry{Grassland (M)}
    
\end{loglogaxis}
\end{tikzpicture}

%% file: pgfplots/owf.tex
\begin{tikzpicture}
\begin{loglogaxis}[
    tiny,
    legend style={nodes={scale=0.6, transform shape}},
    width=\textwidth,
    xlabel={Number of Constraints},
    xlabel shift= -2pt,
    ylabel={Time Cost (s)},
    ylabel shift= -5pt,
    xmin=0, xmax=40000000,
    ymin=0, ymax=25000,
    xtick={1000,10000,100000,1000000,10000000},
    ytick={0.01,0.1,1,10,100,1000,10000},
    legend pos=north west,
    ymajorgrids=true,
    grid style=dashed,
    mark size=1.5pt,
    scaled ticks=false
]

\addplot[color=black, mark=*]
    coordinates {
    (7502,1.12006666666667)(23503,4.3429)(91012,20.5888)(213560,50.6151)(350024,85.6159)(847120,203.182366666667)(2135060,533.7146)(6860048,1737.5721)(33740240,9286.7887)
    };
    \addlegendentry{Gurobi API}
    
\addplot[color=black, mark=triangle]
    coordinates {
    (7502,0.1366)(23503,0.684)(91012,3.948)(213560,10.809)(350024,17.31)(847120,46.886)(2135060,169.078)(6860048,1155)(33740240,21470.932)
    };
    \addlegendentry{JuMP}

\addplot[color=black, mark=asterisk]
    coordinates {
    (7502,0.178533333333333)(23503,0.7674)(91012,4.92353333333333)(213560,13.2079333333333)(350024,39.9356)(847120,71.1166333333333)(2135060,400.4732)(6860048,374.8971)(33740240,2454.1846)
    };
    \addlegendentry{ZIMPL}
    
\addplot[color=black, mark=o]
    coordinates {
    (7502,0.0350492257060262)(23503,0.118090389454209)(91012,0.700482728177171)(213560,1.72823030502316)(350024,3.22135847730035)(847120,7.80823356909391)(2135060,17.7107)(6860048,65.4594)(33740240,660.369)
    };
    \addlegendentry{AMPL}

\addplot[color=blue, mark=square]
    coordinates {
    (7502,0.0256)(23503,0.0991)(91012,0.691733333333333)(213560,1.79723333333333)(350024,3.0352)(847120,7.37406666666667)(2135060,18.8477)(6860048,65.9022333333333)(33740240,332.2728)
    };
    \addlegendentry{Grassland (S)}

\addplot[color=red, mark=diamond]
    coordinates {
    (7502,0.0197666666666667)(23503,0.0709333333333333)(91012,0.156933333333333)(213560,0.257833333333333)(350024,0.406766666666667)(847120,0.767)(2135060,1.67306666666667)(6860048,5.78446666666667)(33740240,24.7764333333333)
    };
    \addlegendentry{Grassland (M)}
    
\end{loglogaxis}
\end{tikzpicture}

%% file: pgfplots/fmi.tex
\begin{tikzpicture}
\begin{loglogaxis}[
    tiny,
    legend style={nodes={scale=0.6, transform shape}},
    width=\textwidth,
    xlabel={Number of Constraints},
    xlabel shift= -2pt,
    ylabel={Time Cost (s)},
    ylabel shift= -5pt,
    xmin=0, xmax=60000000,
    ymin=0, ymax=11000,
    xtick={1000,10000,100000,1000000,10000000},
    ytick={0.01,0.1,1,10,100,1000,10000},
    legend pos=north west,
    ymajorgrids=true,
    grid style=dashed,
    mark size=1.5pt,
    scaled ticks=false
]

\addplot[color=black, mark=*]
    coordinates {
    (5032,0.709633333333333)(10192,1.41136666666667)(20155,2.74793333333333)(50167,6.96133333333333)(100480,14.2575333333333)(200695,28.1073666666667)(501255,71.2264666666667)(1001992,143.243033333333)(2002216,294.059533333333)(5004160,744.3875)(10004560,1566.7432)(20007720,3238.5885)
    };
    \addlegendentry{Gurobi API}
    
\addplot[color=black, mark=triangle]
    coordinates {
    (5032,0.1753)(10192,0.3763)(20155,0.963)(50167,2.614)(100480,5.72)(200695,12.249)(501255,30.946)(1001992,76.112)(2002216,185.558)(5004160,789.857)(10004560,2832.553)(20007720,10334.378)
    };
    \addlegendentry{JuMP}

\addplot[color=black, mark=asterisk]
    coordinates {
    (5032,0.297966666666667)(10192,0.650266666666667)(20155,1.28026666666667)(50167,3.51716666666667)(100480,7.27226666666667)(200695,14.8980666666667)(501255,38.9097)(1001992,80.7829666666667)(2002216,149.5566)(5004160,399.163466666667)(10004560,922.2878)(20007720,2015.3103)
    };
    \addlegendentry{ZIMPL}
    
\addplot[color=black, mark=o]
    coordinates {
    (5032,0.036126)(10192,0.064027)(20155,0.115419)(50167,0.268759)(100480,0.533533)(200695,1.0746)(501255,2.72504)(1001992,5.71547)(2002216,11.8324)(5004160,31.6536)(10004560,87.318)(20007720,256.146)
    };
    \addlegendentry{AMPL}

\addplot[color=blue, mark=square]
    coordinates {
    (5032,0.0771666666666667)(10192,0.1567)(20155,0.317733333333333)(50167,0.8271)(100480,1.4799)(200695,2.8593)(501255,7.0362)(1001992,14.847)(2002216,30.4616)(5004160,80.8257)(10004560,167.8728)(20007720,343.2538)
    };
    \addlegendentry{Grassland (S)}

\addplot[color=red, mark=diamond]
    coordinates {
    (5032,0.0569)(10192,0.0952666666666667)(20155,0.193533333333333)(50167,0.1437)(100480,0.2594)(200695,0.4305)(501255,0.6131)(1001992,1.0784)(2002216,1.9496)(5004160,5.2839)(10004560,9.169)(20007720,17.5701)
    };
    \addlegendentry{Grassland (M)}
    
\end{loglogaxis}
\end{tikzpicture}

%% file: pgfplots/decomposition.tex
\begin{tikzpicture}
	\begin{axis}[
    	tiny,
        legend style={nodes={scale=0.75, transform shape}},
        width=\textwidth,
        xlabel={Methods},
        xlabel shift= -2pt,
        xtick=data,
        xticklabels={Baseline,RH,FRH,G-RH,G-FRH},
        ylabel={Time Cost (s)},
        ylabel shift= -5pt,
	    ybar stacked,
	    ymin=0,
	    xtick pos=left,
	    ytick pos=left
	]
	\addplot coordinates
		{(0,856.995) (1,3.12) (2,15.32) (3,43.83) (4,49.49)};
		\addlegendentry{Instantiation time}
	\addplot coordinates
		{(0,3135.16) (1,40.80) (2,67.68) (3,73.31) (4,167.62)};
		\addlegendentry{Solving time}
	\addlegendimage{/pgfplots/refstyle={plt:sd_obj}}
    \addlegendentry{Objective value}
    \addlegendimage{/pgfplots/refstyle={plt:sd_optimal_obj}}
    \addlegendentry{Optimal obj}
	
	\end{axis}
	\begin{axis}[
    	tiny,
        legend style={nodes={scale=0.75, transform shape}},
        width=\textwidth,
        axis y line*=right,
        xtick=data,
        xticklabels={,,,,},
        ylabel={Objective Value},
        ylabel shift= -5pt,
	    ymin=5.1e9, ymax=5.65e9,
	    nodes near coords,
	    point meta=explicit symbolic,
	   	every node near coord/.append style={
	   	    font=\tiny,
            align=center,
            text width=1cm
        }
	]
	\addplot[sharp plot,line legend,color=black, mark=*] coordinates
		{(0,5.15626e9) [$5.156 \cdot 10^9$]
		(1,5.54017e9) [$5.540 \cdot 10^9$ (+7.45\%)]
		(2,5.19818e9) [$5.198 \cdot 10^9$ (+0.81\%)]
		(3,5.26513e9) [$5.265 \cdot 10^9$ (+2.11\%)]
		(4,5.17477e9) [$5.175 \cdot 10^9$ (+0.36\%)]} 
		;\label{plt:sd_obj}
	\addplot[sharp plot,line legend,black,update limits=false, dashed] 
    	coordinates {(-1,5.15626e9) (5,5.15626e9)};\label{plt:sd_optimal_obj}
	\end{axis}
\end{tikzpicture}

%% file: pgfplots/decomposition_fine_tuning.tex
\begin{tikzpicture}
	\begin{axis}[
    	tiny,
        legend style={nodes={scale=0.75, transform shape}},
        width=\textwidth,
        xlabel={Methods},
        xlabel shift= -2pt,
        xtick=data,
        xticklabels={Baseline,RH,FRH,G-RH,G-FRH},
        ylabel={Time Cost (s)},
        ylabel shift= -5pt,
	    ybar stacked,
	    ymin=0,
	    xtick pos=left,
	    ytick pos=left
	]
	\addplot coordinates
		{(0,856.995) (1,8.16) (2,17.73) (3,46.92) (4,51.99)};
		\addlegendentry{Instantiation time}
	\addplot coordinates
		{(0,3135.16) (1,191.15) (2,214.79) (3,211.42) (4,284.03)};
		\addlegendentry{Solving time}
	\addlegendimage{/pgfplots/refstyle={plt:sd_obj_f}}
    \addlegendentry{Objective value}
    \addlegendimage{/pgfplots/refstyle={plt:sd_optimal_obj_f}}
    \addlegendentry{Optimal obj}
	\end{axis}
	\begin{axis}[
    	tiny,
        legend style={nodes={scale=0.5, transform shape}},
        width=\textwidth,
        axis y line*=right,
        xtick=data,
        xticklabels={,,,,},
	    ymin=5.1e9, ymax=5.65e9,
	    nodes near coords,
	    point meta=explicit symbolic,
	   	every node near coord/.append style={
	   	    font=\tiny,
            align=center,
            text width=1cm
        },
	    ylabel={Objective Value},
        ylabel shift= -5pt,
	]
	\addplot[sharp plot,line legend,color=black, mark=*] coordinates
		{(0,5.15626e9) [$5.156 \cdot 10^9$]
		(1,5.33775e9) [$5.338 \cdot 10^9$ (+3.52\%)]
		(2,5.19651e9) [$5.197 \cdot 10^9$ (+0.78\%)]
		(3,5.22104e9) [$5.221 \cdot 10^9$ (+1.26\%)]
		(4,5.17321e9) [$5.173 \cdot 10^9$ (+0.33\%)]}
		;\label{plt:sd_obj_f}
	\addplot[sharp plot,line legend,black,update limits=false, dashed] 
    	coordinates {(-1,5.15626e9) (5,5.15626e9)};\label{plt:sd_optimal_obj_f}
	\end{axis}
\end{tikzpicture}

%% file: pgfplots/sd.tex
\begin{tikzpicture}
\begin{semilogyaxis}[
    tiny,
    legend style={nodes={scale=0.65, transform shape}},
    width=\textwidth,
    xlabel={Number of Constraints},
    xlabel shift= -2pt,
    ylabel={Time Cost (s)},
    ylabel shift= -5pt,
    xmin=0, xmax=4000000,
    ymin=0, ymax=80000,
    xtick={0,1000000,2000000,3000000,4000000},
    ytick={1,10,100,1000,10000},
    legend pos=north west,
    ymajorgrids=true,
    grid style=dashed,
    mark size=1.5pt,
    scaled ticks=false
]

\addplot[color=black, mark=*]
    coordinates {
    (427414,115.5568)(868668,268.7541)(1311685,481.0783)(1755291,893.3147)(2199175,1139.4417)(2643484,1277.4025)(3087902,1631.6566)(3443465,2445.5923)
    };
    \addlegendentry{ZIMPL based version}
    
\addplot[color=black, mark=asterisk]
    coordinates {
    (427414,54.4043)(868668,112.8448)(1311685,174.8516)(1755291,232.9484)(2199175,295.835)(2643484,359.1428)(3087902,423.8897)(3443465,482.1621)
    };
    \addlegendentry{Gurobi API based version}

\addplot[color=blue, mark=square]
    coordinates {
    (427414,15.4386)(868668,30.8101)(1311685,46.6049)(1755291,62.179)(2199175,79.0009)(2643484,94.8718)(3087902,108.5973)(3443465,126.6095)
    };
    \addlegendentry{Grassland (S)}

\addplot[color=red, mark=diamond]
    coordinates {
    (427414,1.3707)(868668,1.9367)(1311685,2.2098)(1755291,2.7032)(2199175,2.9862)(2643484,3.7799)(3087902,4.7175)(3443465,6.2102)
    };
    \addlegendentry{Grassland (M)}
    
\end{semilogyaxis}
\end{tikzpicture}

%% file: main_cikm_arxiv.bbl

\begin{thebibliography}{19}


\ifx \showCODEN    \undefined \def \showCODEN     #1{\unskip}     \fi
\ifx \showDOI      \undefined \def \showDOI       #1{#1}\fi
\ifx \showISBNx    \undefined \def \showISBNx     #1{\unskip}     \fi
\ifx \showISBNxiii \undefined \def \showISBNxiii  #1{\unskip}     \fi
\ifx \showISSN     \undefined \def \showISSN      #1{\unskip}     \fi
\ifx \showLCCN     \undefined \def \showLCCN      #1{\unskip}     \fi
\ifx \shownote     \undefined \def \shownote      #1{#1}          \fi
\ifx \showarticletitle \undefined \def \showarticletitle #1{#1}   \fi
\ifx \showURL      \undefined \def \showURL       {\relax}        \fi
\providecommand\bibfield[2]{#2}
\providecommand\bibinfo[2]{#2}
\providecommand\natexlab[1]{#1}
\providecommand\showeprint[2][]{arXiv:#2}

\bibitem[\protect\citeauthoryear{Abadi, Barham, Chen, Chen, Davis, Dean, Devin,
  Ghemawat, Irving, Isard, Kudlur, Levenberg, Monga, Moore, Murray, Steiner,
  Tucker, Vasudevan, Warden, Wicke, Yu, and Zheng}{Abadi et~al\mbox{.}}{2016}]%
        {tensorflow}
\bibfield{author}{\bibinfo{person}{Mart{\'\i}n Abadi}, \bibinfo{person}{Paul
  Barham}, \bibinfo{person}{Jianmin Chen}, \bibinfo{person}{Zhifeng Chen},
  \bibinfo{person}{Andy Davis}, \bibinfo{person}{Jeffrey Dean},
  \bibinfo{person}{Matthieu Devin}, \bibinfo{person}{Sanjay Ghemawat},
  \bibinfo{person}{Geoffrey Irving}, \bibinfo{person}{Michael Isard},
  \bibinfo{person}{Manjunath Kudlur}, \bibinfo{person}{Josh Levenberg},
  \bibinfo{person}{Rajat Monga}, \bibinfo{person}{Sherry Moore},
  \bibinfo{person}{Derek~G. Murray}, \bibinfo{person}{Benoit Steiner},
  \bibinfo{person}{Paul Tucker}, \bibinfo{person}{Vijay Vasudevan},
  \bibinfo{person}{Pete Warden}, \bibinfo{person}{Martin Wicke},
  \bibinfo{person}{Yuan Yu}, {and} \bibinfo{person}{Xiaoqiang Zheng}.}
  \bibinfo{year}{2016}\natexlab{}.
\newblock \showarticletitle{TensorFlow: A System for Large-Scale Machine
  Learning}. In \bibinfo{booktitle}{\emph{12th {USENIX} Symposium on Operating
  Systems Design and Implementation ({OSDI} 16)}}. \bibinfo{publisher}{{USENIX}
  Association}, \bibinfo{address}{Savannah, GA}, \bibinfo{pages}{265--283}.
\newblock
\showISBNx{978-1-931971-33-1}
\urldef\tempurl%
\url{https://www.usenix.org/conference/osdi16/technical-sessions/presentation/abadi}
\showURL{%
\tempurl}


\bibitem[\protect\citeauthoryear{Adams}{Adams}{2018}]%
        {ryu}
\bibfield{author}{\bibinfo{person}{Ulf Adams}.}
  \bibinfo{year}{2018}\natexlab{}.
\newblock \showarticletitle{Ry\={u}: Fast Float-to-String Conversion}. In
  \bibinfo{booktitle}{\emph{Proceedings of the 39th ACM SIGPLAN Conference on
  Programming Language Design and Implementation}} (Philadelphia, PA, USA)
  \emph{(\bibinfo{series}{PLDI 2018})}. \bibinfo{publisher}{Association for
  Computing Machinery}, \bibinfo{address}{New York, NY, USA},
  \bibinfo{pages}{270–282}.
\newblock
\showISBNx{9781450356985}
\urldef\tempurl%
\url{https://doi.org/10.1145/3192366.3192369}
\showDOI{\tempurl}


\bibitem[\protect\citeauthoryear{ApS}{ApS}{2019}]%
        {mosek}
\bibfield{author}{\bibinfo{person}{MOSEK ApS}.}
  \bibinfo{year}{2019}\natexlab{}.
\newblock \bibinfo{title}{MOSEK optimization suite}.
\newblock
\newblock


\bibitem[\protect\citeauthoryear{Bradley}{Bradley}{1977}]%
        {BradleyStephenP.1977Amp}
\bibfield{author}{\bibinfo{person}{Stephen~P. Bradley}.}
  \bibinfo{year}{1977}\natexlab{}.
\newblock \bibinfo{booktitle}{\emph{Applied mathematical programming}}.
\newblock \bibinfo{publisher}{Addison-Wesley Pub. Co.},
  \bibinfo{address}{Reading, Mass.}
\newblock
\showISBNx{020100464X}
\showLCCN{^^^76010426^}


\bibitem[\protect\citeauthoryear{Brook, Kendrick, and Meeraus}{Brook
  et~al\mbox{.}}{1988}]%
        {brook1988gams}
\bibfield{author}{\bibinfo{person}{Anthony Brook}, \bibinfo{person}{David
  Kendrick}, {and} \bibinfo{person}{Alexander Meeraus}.}
  \bibinfo{year}{1988}\natexlab{}.
\newblock \showarticletitle{GAMS, a user's guide}.
\newblock \bibinfo{journal}{\emph{ACM Signum Newsletter}} \bibinfo{volume}{23},
  \bibinfo{number}{3-4} (\bibinfo{year}{1988}), \bibinfo{pages}{10--11}.
\newblock


\bibitem[\protect\citeauthoryear{Chen and Wang}{Chen and Wang}{1997}]%
        {chen1997linear}
\bibfield{author}{\bibinfo{person}{Mingyuan Chen} {and} \bibinfo{person}{Weimin
  Wang}.} \bibinfo{year}{1997}\natexlab{}.
\newblock \showarticletitle{A linear programming model for integrated steel
  production and distribution planning}.
\newblock \bibinfo{journal}{\emph{International Journal of Operations \&
  Production Management}} (\bibinfo{year}{1997}).
\newblock


\bibitem[\protect\citeauthoryear{Cornuejols and T{\"u}t{\"u}nc{\"u}}{Cornuejols
  and T{\"u}t{\"u}nc{\"u}}{2006}]%
        {cornuejols2006optimization}
\bibfield{author}{\bibinfo{person}{Gerard Cornuejols} {and}
  \bibinfo{person}{Reha T{\"u}t{\"u}nc{\"u}}.} \bibinfo{year}{2006}\natexlab{}.
\newblock \bibinfo{booktitle}{\emph{Optimization methods in finance}}.
  Vol.~\bibinfo{volume}{5}.
\newblock \bibinfo{publisher}{Cambridge University Press}.
\newblock


\bibitem[\protect\citeauthoryear{Dimitriadis, Shah, and Pantelides}{Dimitriadis
  et~al\mbox{.}}{1997}]%
        {DIMITRIADIS1997S1061}
\bibfield{author}{\bibinfo{person}{A.D. Dimitriadis}, \bibinfo{person}{N.
  Shah}, {and} \bibinfo{person}{C.C. Pantelides}.}
  \bibinfo{year}{1997}\natexlab{}.
\newblock \showarticletitle{RTN-based rolling horizon algorithms for medium
  term scheduling of multipurpose plants}.
\newblock \bibinfo{journal}{\emph{Computers \& Chemical Engineering}}
  \bibinfo{volume}{21} (\bibinfo{year}{1997}), \bibinfo{pages}{S1061 -- S1066}.
\newblock
\showISSN{0098-1354}


\bibitem[\protect\citeauthoryear{Epstein, Neely, Weintraub, Valenzuela,
  Hurtado, Gonzalez, Beiza, Naveas, Infante, Alarcon, et~al\mbox{.}}{Epstein
  et~al\mbox{.}}{2012}]%
        {epstein2012strategic}
\bibfield{author}{\bibinfo{person}{Rafael Epstein}, \bibinfo{person}{Andres
  Neely}, \bibinfo{person}{Andres Weintraub}, \bibinfo{person}{Fernando
  Valenzuela}, \bibinfo{person}{Sergio Hurtado}, \bibinfo{person}{Guillermo
  Gonzalez}, \bibinfo{person}{Alex Beiza}, \bibinfo{person}{Mauricio Naveas},
  \bibinfo{person}{Florencio Infante}, \bibinfo{person}{Fernando Alarcon},
  {et~al\mbox{.}}} \bibinfo{year}{2012}\natexlab{}.
\newblock \showarticletitle{A strategic empty container logistics optimization
  in a major shipping company}.
\newblock \bibinfo{journal}{\emph{Interfaces}} \bibinfo{volume}{42},
  \bibinfo{number}{1} (\bibinfo{year}{2012}), \bibinfo{pages}{5--16}.
\newblock


\bibitem[\protect\citeauthoryear{Fourer, Gay, and Kernighan}{Fourer
  et~al\mbox{.}}{1990}]%
        {fourer1990modeling}
\bibfield{author}{\bibinfo{person}{Robert Fourer}, \bibinfo{person}{David~M
  Gay}, {and} \bibinfo{person}{Brian~W Kernighan}.}
  \bibinfo{year}{1990}\natexlab{}.
\newblock \showarticletitle{A modeling language for mathematical programming}.
\newblock \bibinfo{journal}{\emph{Management Science}} \bibinfo{volume}{36},
  \bibinfo{number}{5} (\bibinfo{year}{1990}), \bibinfo{pages}{519--554}.
\newblock


\bibitem[\protect\citeauthoryear{{Gurobi Optimization, LLC}}{{Gurobi
  Optimization, LLC}}{2018}]%
        {gurobi2018gurobi}
\bibfield{author}{\bibinfo{person}{{Gurobi Optimization, LLC}}.}
  \bibinfo{year}{2018}\natexlab{}.
\newblock \showarticletitle{Gurobi optimizer reference manual}.
\newblock
  \bibinfo{howpublished}{\url{https://www.gurobi.com/documentation/9.0/refman/index.html}}.
\newblock  (\bibinfo{year}{2018}).
\newblock


\bibitem[\protect\citeauthoryear{{Gurobi Optimization, LLC}}{{Gurobi
  Optimization, LLC}}{2019}]%
        {gurobi_modeling_examples}
\bibfield{author}{\bibinfo{person}{{Gurobi Optimization, LLC}}.}
  \bibinfo{year}{2019}\natexlab{}.
\newblock \bibinfo{title}{Gurobi modeling examples}.
\newblock
  \bibinfo{howpublished}{\url{https://gurobi.github.io/modeling-examples/}}.
\newblock


\bibitem[\protect\citeauthoryear{Hart, Watson, and Woodruff}{Hart
  et~al\mbox{.}}{2011}]%
        {hart2011pyomo}
\bibfield{author}{\bibinfo{person}{William~E Hart}, \bibinfo{person}{Jean-Paul
  Watson}, {and} \bibinfo{person}{David~L Woodruff}.}
  \bibinfo{year}{2011}\natexlab{}.
\newblock \showarticletitle{Pyomo: modeling and solving mathematical programs
  in Python}.
\newblock \bibinfo{journal}{\emph{Mathematical Programming Computation}}
  \bibinfo{volume}{3}, \bibinfo{number}{3} (\bibinfo{year}{2011}),
  \bibinfo{pages}{219}.
\newblock


\bibitem[\protect\citeauthoryear{Koch}{Koch}{2004}]%
        {Koch2004}
\bibfield{author}{\bibinfo{person}{Thorsten Koch}.}
  \bibinfo{year}{2004}\natexlab{}.
\newblock \emph{\bibinfo{title}{Rapid Mathematical Programming}}.
\newblock \bibinfo{thesistype}{Ph.D. Dissertation}. \bibinfo{school}{Technische
  {Universit\"at} Berlin}.
\newblock
\urldef\tempurl%
\url{http://www.zib.de/Publications/abstracts/ZR-04-58/}
\showURL{%
\tempurl}
\newblock
\shownote{ZIB-Report 04-58.}


\bibitem[\protect\citeauthoryear{Lofberg}{Lofberg}{2004}]%
        {lofberg2004yalmip}
\bibfield{author}{\bibinfo{person}{Johan Lofberg}.}
  \bibinfo{year}{2004}\natexlab{}.
\newblock \showarticletitle{YALMIP: A toolbox for modeling and optimization in
  MATLAB}. In \bibinfo{booktitle}{\emph{2004 IEEE international conference on
  robotics and automation (IEEE Cat. No. 04CH37508)}}. IEEE,
  \bibinfo{pages}{284--289}.
\newblock


\bibitem[\protect\citeauthoryear{Lubin and Dunning}{Lubin and Dunning}{2015}]%
        {LubinDunningIJOC}
\bibfield{author}{\bibinfo{person}{Miles Lubin} {and} \bibinfo{person}{Iain
  Dunning}.} \bibinfo{year}{2015}\natexlab{}.
\newblock \showarticletitle{Computing in Operations Research Using Julia}.
\newblock \bibinfo{journal}{\emph{INFORMS Journal on Computing}}
  \bibinfo{volume}{27}, \bibinfo{number}{2} (\bibinfo{year}{2015}),
  \bibinfo{pages}{238--248}.
\newblock
\urldef\tempurl%
\url{https://doi.org/10.1287/ijoc.2014.0623}
\showDOI{\tempurl}


\bibitem[\protect\citeauthoryear{Paszke, Gross, Massa, Lerer, Bradbury, Chanan,
  Killeen, Lin, Gimelshein, Antiga, Desmaison, Kopf, Yang, DeVito, Raison,
  Tejani, Chilamkurthy, Steiner, Fang, Bai, and Chintala}{Paszke
  et~al\mbox{.}}{2019}]%
        {PyTorch}
\bibfield{author}{\bibinfo{person}{Adam Paszke}, \bibinfo{person}{Sam Gross},
  \bibinfo{person}{Francisco Massa}, \bibinfo{person}{Adam Lerer},
  \bibinfo{person}{James Bradbury}, \bibinfo{person}{Gregory Chanan},
  \bibinfo{person}{Trevor Killeen}, \bibinfo{person}{Zeming Lin},
  \bibinfo{person}{Natalia Gimelshein}, \bibinfo{person}{Luca Antiga},
  \bibinfo{person}{Alban Desmaison}, \bibinfo{person}{Andreas Kopf},
  \bibinfo{person}{Edward Yang}, \bibinfo{person}{Zachary DeVito},
  \bibinfo{person}{Martin Raison}, \bibinfo{person}{Alykhan Tejani},
  \bibinfo{person}{Sasank Chilamkurthy}, \bibinfo{person}{Benoit Steiner},
  \bibinfo{person}{Lu Fang}, \bibinfo{person}{Junjie Bai}, {and}
  \bibinfo{person}{Soumith Chintala}.} \bibinfo{year}{2019}\natexlab{}.
\newblock \showarticletitle{PyTorch: An Imperative Style, High-Performance Deep
  Learning Library}. In \bibinfo{booktitle}{\emph{Advances in Neural
  Information Processing Systems}},
  \bibfield{editor}{\bibinfo{person}{H.~Wallach},
  \bibinfo{person}{H.~Larochelle}, \bibinfo{person}{A.~Beygelzimer},
  \bibinfo{person}{F.~d\textquotesingle Alch\'{e}-Buc},
  \bibinfo{person}{E.~Fox}, {and} \bibinfo{person}{R.~Garnett}} (Eds.),
  Vol.~\bibinfo{volume}{32}. \bibinfo{publisher}{Curran Associates, Inc.},
  \bibinfo{pages}{8026--8037}.
\newblock
\urldef\tempurl%
\url{https://proceedings.neurips.cc/paper/2019/file/bdbca288fee7f92f2bfa9f7012727740-Paper.pdf}
\showURL{%
\tempurl}


\bibitem[\protect\citeauthoryear{Pochet and Wolsey}{Pochet and Wolsey}{2010}]%
        {10.5555/1951836}
\bibfield{author}{\bibinfo{person}{Yves Pochet} {and}
  \bibinfo{person}{Laurence~A. Wolsey}.} \bibinfo{year}{2010}\natexlab{}.
\newblock \bibinfo{booktitle}{\emph{Production Planning by Mixed Integer
  Programming} (\bibinfo{edition}{1st} ed.)}.
\newblock \bibinfo{publisher}{Springer Publishing Company, Incorporated}.
\newblock
\showISBNx{144192132X}


\bibitem[\protect\citeauthoryear{Sethi and Sorger}{Sethi and Sorger}{1991}]%
        {Sethi1991}
\bibfield{author}{\bibinfo{person}{Suresh Sethi} {and} \bibinfo{person}{Gerhard
  Sorger}.} \bibinfo{year}{1991}\natexlab{}.
\newblock \showarticletitle{A theory of rolling horizon decision making}.
\newblock \bibinfo{journal}{\emph{Annals of Operations Research}}
  \bibinfo{volume}{29}, \bibinfo{number}{1} (\bibinfo{date}{01 Dec}
  \bibinfo{year}{1991}), \bibinfo{pages}{387--415}.
\newblock
\showISSN{1572-9338}
\urldef\tempurl%
\url{https://doi.org/10.1007/BF02283607}
\showDOI{\tempurl}


\end{thebibliography}
